\documentclass[twocolumn,superscriptaddress,showpacs,nofootinbib,notitlepage,preprintnumbers,secnumarabic,amssymb, nobibnotes, aps, prd]{revtex4-2}
\usepackage[utf8]{inputenc}
\usepackage{graphicx}
\usepackage{latexsym,amsmath,amssymb,amsthm,lmodern,float,url,bm}
\usepackage{natbib}
\usepackage{color}
\usepackage{microtype}
\usepackage{import}
\usepackage{bbold}
\usepackage[dvipsnames]{xcolor}
\usepackage[plain]{fancyref}
\usepackage{varioref}
\usepackage{slashed}
\usepackage{multirow}
\usepackage{tikz}
\usepackage{physics}
\usepackage{comment}
\usetikzlibrary{backgrounds}
\usetikzlibrary{arrows,shapes}
\usetikzlibrary{tikzmark}
\usetikzlibrary{calc}
\usetikzlibrary{positioning}
\usepackage[normalem]{ulem}
\usepackage{mathtools, nccmath}
\usepackage{wrapfig}
\usepackage{comment}
\usepackage{bbm}
\usepackage{adjustbox}
\usepackage{soul}
\setstcolor{purple}

\renewcommand{\eq}[1]{Eq.~(\ref{eq:#1})}

\newtheorem{proposition}{Proposition}[section]

\usepackage[colorlinks=true,backref=false, linktocpage=true,
citecolor=blue,urlcolor=blue,linkcolor=blue,pdfpagemode=UseOutlines]{hyperref}
\usepackage{cleveref}

\hypersetup{%
 bookmarksnumbered=true,
 pdftitle = {},
 pdfsubject = {},
 pdfauthor = {},
 pdfkeywords = {}
}
\usepackage{orcidlink}
\usepackage{todonotes}

\Crefname{appendix}{App.}{Apps.}
\Crefname{equation}{Eq.}{Eqs.}
\Crefname{figure}{Fig.}{Figs.}

\newtheorem{theorem}{Theorem}

\newtheorem{lemma}{Lemma}

\newtheorem{definition}{Definition}

\usepackage{silence}
\newcommand{\mjc}[1]{\textcolor{teal}{#1}}

\newcommand{\sqms}{\affiliation{Superconducting and Quantum Materials System Center (SQMS), Batavia, IL, 60510, USA.}}
\newcommand{\fnal}{\affiliation{Fermi National Accelerator Laboratory, 
Batavia, IL 60510, USA}}

\begin{document}
 \preprint{FERMILAB-PUB-25-0709-SQMS-T, NT@UW-25-19}
 \title{Magic State Distillation using Asymptotically Good Codes on Qudits}
 \author{Michael J.~Cervia\,
 \orcidlink{0000-0002-2962-3055}
 }
 \email{cervia@uw.edu}
 \affiliation{Department of Physics, University of Washington, Seattle, WA, 98195, USA}
\author{Henry Lamm\,\orcidlink{0000-0003-3033-0791}}
\email{hlamm@fnal.gov}
\sqms
\fnal
\author{Diyi Liu\,\orcidlink{0000-0002-0996-7686}}
\email{diyiliu@lbl.gov}
\affiliation{Applied Mathematics and Computational Research Division, Lawrence Berkeley National Laboratory, Berkeley, CA 94720, USA}
\author{Edison M. Murairi\,\orcidlink{0000-0002-1639-6308}}
\email{emurairi@fnal.gov}
\sqms
\fnal
\author{Shuchen Zhu\,\orcidlink{0000-0002-1240-2002}}
\email{shuchen.zhu@duke.edu} 
\affiliation{Department of Mathematics, Duke University, Durham, NC 27708, USA}
\affiliation{Duke Quantum Center, Duke University, Durham, NC 27701, USA}
\date{\today}

\begin{abstract}
Qudits offer the potential for low-overhead magic state distillation, although previous results for asymptotically good codes have required qudit dimension $q\gg 100$ or code length $\mathcal{N}\gg 100$.  These parameters far exceed experimental demonstrations of qudit platforms, and thus motivate the search for better codes.  
Using a novel lifting procedure, we construct the first family of good triorthogonal codes on the $\mathbb{F}_{2^{2m}}$ alphabet with $m \geq 3$ that lies above the Tsfasman-Vladut-Zink bound. These codes yield a family of asymptotically good quantum codes with transversal CCZ gates, enabling constant space overhead magic state distillation with qudit dimension as small as $q=64$. Further, we identify a promising code with parameters $[[42,14,6]]_{64}$. Finally, we show that a distilled $\ket{CCZ}_{2^{2m}}$ can be reduced to a $\ket{CCZ}_{2^n}$ state for arbitrary $n$ with a constant-depth Clifford circuit of at most 9 computational basis measurements, 12 single-qudit and 9 two-qudit Clifford gates.
\end{abstract}
\maketitle

\section{Introduction}\label{sec:introduction}
Quantum computation with higher-dimensional systems, or qudits, has emerged as a compelling alternative to qubit-based architectures~\cite{wang2020qudits, KiktenkoNikolaevaEtAl2025, nguyen2024empowering}.
The enlarged Hilbert space per quantum register leads to higher effective connectivity by replacing entangling gates with single-qudit rotations~\cite{Nikolaeva:2021rhq,Mansky:2022bai,Jankovic:2023bnw,Murairi:2024xpc,Keppens:2025pfo}. This alternative permits lower gate fidelities for the same algorithmic fidelity~\cite{Nikolaeva:2022wmq,Iiyama:2024uos,Champion:2024ufp,kiktenko2020scalable, roy2023two,PRXQuantum.2.020311, zhu2024unified}, which in turn would yield higher error thresholds for the prospect of fault-tolerant computation. This promise has driven algorithmic research in qudits across multiple fields, including
numerical optimization~\cite{Tancara:2024vck,Bravyi:2020pur}, 
condensed matter~\cite{PhysRevA.110.062602,choi2017dynamical,BassmanOftelie:2022hfz,Young:2023zxu}, and particle physics~\cite{Gustafson:2021qbt,Gonzalez-Cuadra:2022hxt,Popov:2023xft,Calajo:2024qrc,Illa:2024kmf, Zache:2023cfj,Meth:2023wzd,Kurkcuoglu:2024cfv,Joshi:2025pgv}.

Complementing these theoretical efforts, recent experiments on various physical platforms have demonstrated both the capability and near-term potential of large-dimension qudit quantum processors. In trapped-ion systems, researchers have realized qudits with dimensions $q=7$~\cite{Ringbauer:2021lhi}, $q=8$~\cite{Low:2023dlg}, and $q=13$~\cite{Shi:2025vvq}. Transmon systems have further realized $q=8$~\cite{Champion:2024wlh} and $q=12$~\cite{Wang:2024xbz}, illustrating that the circuit quantum electrodynamics naturally extends to larger local dimensions. 
Molecular spins ($q=8$~\cite{Mezzadri:2023ige}), solid-state spin systems 
($q=10$~\cite{Omanakuttan:2021ffn}, $q=16$~\cite{deFuentes:2023pbp}) 
and cold-atom experiments (reaching $q=25$~\cite{Dong:2023xhv}) further demonstrate how diverse hardware can support many-level encoding with promising coherence and control properties. Notably, superconducting radio-frequency (SRF) cavities have demonstrated extremely large single-mode Hilbert spaces ($q=20$ with high fidelity~\cite{Kim:2025ywx} and up to $q=100$ with lower fidelity~\cite{Deng:2023uac}). Altogether, these experiments point toward architectures where large local dimensions are available. 

Fault-tolerant computations on these architectures requires new quantum error correcting codes and protocols. Magic state distillation (MSD) is commonly proposed to enable universal computation by preparing high-fidelity resource states from several low-fidelity ones~\cite{BravyiKitaev05}. The resulting resource state is used to implement a non-Clifford gate to a desired accuracy. The number of noisy inputs required to achieve an infidelity $\epsilon$ scales as $\sim\log^{\gamma}(1/\epsilon)$, where the exponent $\gamma$ is called the MSD overhead. 

In Ref.~\cite{PhysRevA.86.052329}, it was shown that a quantum code with a transversal non-Clifford gate may be used to distill the corresponding resource with overhead equal to 
\begin{align}
    \gamma = \frac{\log{(\mathcal{N}/\mathcal{K})}}{\log \mathcal{D}},
\end{align}
where $\mathcal{N}$, $\mathcal{K}$, and $\mathcal{D}$ are the code length, code dimension, and code distance of a quantum code $[[\mathcal{N},\mathcal{K},\mathcal{D}]]_q$ over $q$-dimensional qudit. The authors identified qubit codes achieving $\gamma = 1.585$ and conjectured a lower bound $\gamma \geq 1$. Their construction relied on classical triorthogonal codes. 

Since then, several efforts have been made to construct quantum codes with transversal non-Clifford gates and lower overhead $\gamma$. Quantum Reed–Muller codes over prime dimensions allow implementation of some transversal non-Clifford gates and $\gamma<1$~\cite{Campbell:2012olh,HowardVala2012,Saha:2025frb}.
Another code with $\gamma \leq 0.678$ was found to be achievable for $q \geq 2^{58}$~\cite{HastingsHaah2018}. Further,  Calderbank-Shor-Steane (CSS) codes constructed from Reed-Solomon codes could obtain arbitrary small overhead ($\gamma \to 0$) as $q\rightarrow\infty$~\cite{PhysRevLett.123.070507}. 
Recently, it was demonstrated that constant space overhead and constant overhead are achievable with finite, albeit large, $q\geq 2^6$ and $q \geq 2^{10}$ respectively by constructing \emph{asymptotically good codes} which are families of codes which as the code length increases the encoding rate and relative distance are bounded away from zero~\cite{Wills:2024wid,Nguyen:2024qwg,GolowichGuruswami2024}. 

The existence of these codes naturally raises the question of how much more efficient can the codes be. This can partially be addressed by considering  the efficiency of constituent classical codes $[n,k,d]_q$ which can be quantified by how large the encoding rate $R=\frac{k}{n}$ can be for relative distance $\delta=\frac{d}{n}$, usually with respect to two lower bounds.  First is the Gilbert--Varshamov (GV) bound~\cite{6773017,Varshamov1957} provides a non-constructive lower bound based on sphere-packing arguments. 
Remarkably, randomly-drawn linear code yields parameters that saturate the GV bound with high probability~\cite{agcodes-advanced}, and therefore more efficient codes should rely upon additional structure. When $q\geq49$, the Tsfasman--Vladut--Zink (TVZ) bound~\cite{mana.19821090103} is tighter than the GV bound on an interval around $\delta\sim0.4$ which grows larger with $q$. This latter bound is the one we will consider here since for all our codes $q\geq 64$. 
Codes meeting or exceeding the TVZ bound have large $R$ for their $\delta$, 
a feature desirable for optimizing quantum-code performance. While previous works~\cite{Wills:2024wid,Nguyen:2024qwg} did not demonstrate it, Garcia--Stichtenoth towers can yield families of asymptotically good codes that satisfy the 
TVZ bound~\cite{mana.19821090103}; we refer to such codes as TVZ codes. While no triorthogonal TVZ code has been previous identified, TVZ codes with iso-duality (a less-restrictive condition) exist~\cite{chara2025goodisodualagcodestowers}. 

In this work, we present a construction of triorthogonal TVZ codes over  
$q = r^2$, where $r \ge 8$ is a power of two. Our approach is based on a \emph{lifting method}, 
which we prove preserves triorthogonality from a base code. Code lifting is a well-known technique to create new codes from old codes; see e.g. Refs.~\cite{Chara2024,chara2025goodisodualagcodestowers,m-gretchen,DOUGHERTY2005123,10931129,Old_2024}. 
In the context of algebraic geometry (AG) codes, a similar lifting method was employed~\cite{Chara2024,chara2025goodisodualagcodestowers} to establish the existence of iso-dual TVZ codes~\cite{chara2025goodisodualagcodestowers}. 
We adapt and extend the lifting procedure to maintain triorthogonality, and obtain triorthogonal TVZ codes. Thus, we find that despite the strong algebraic constraints of triorthogonality, codes exist with `very good' parameters.  We further show how these classical codes can be used to construct quantum codes with correspondingly improved MSD performance. 

The rest of this article is organized as follows: Sec.~\ref{sec:summary-results} summarizes our results. In Sec.~\ref{sec:basic-notions}, we review the notions of algebraic geometry, defining AG codes and their parameters, with particular focus on triorthogonal codes. 
In Sec.~\ref{sec:classical-code} we provide our construction of triorthogonal TVZ codes, and in Sec.~\ref{sec:quantum-codes} we derive a corresponding family of good quantum codes, explicit examples of which we provide in Sec.~\ref{sec:explicit-codes}. 
Furthermore, in Sec.~\ref{sec:msd}, we show how to obtain $\ket{CCZ}$ magic states for any $2^n$-qudit. 
Finally, we summarize and reflect upon future directions in Sec.~\ref{sec:conclusion}.

\section{Summary of Results}\label{sec:summary-results}

Our approach to developing quantum codes introduces a lifting method that generates larger triorthogonal codes from smaller ones, while preserving the triorthogonality. 
Using our lifting method, we prove the existence of asymptotically good families of triorthogonal codes defined over the alphabet $\mathbb{F}_q$ with $q=2^{2m}$ and $m\geq 3$, improving on the result of Refs.~\cite{Nguyen:2024qwg} and Ref.~\cite{Wills:2024wid} where it was required $m \geq 4$ and $m \geq 5$, respectively. We further show that our codes can exceed the TVZ bound~\cite{mana.19821090103}. 

Using our triorthogonal codes, we have constructed a family of good quantum codes with transversal CCZ gate on qudit $q = 2^{2m}$ where $m \geq 3$, enabling magic state distillation with constant space $\gamma=0$ overhead following e.g. the protocol in Ref.~\cite{Nguyen:2024qwg}. Further, we have identified a low-overhead $\gamma\approx0.613$ code $[[42,14,6]]_{64}$. Finally, we show that we can obtain the $\ket{CCZ}_{2^n}$ state for any $n$ with a constant-depth Clifford circuit given the ability to distill a $\ket{CCZ}_{2^{2m}}$ ($m \geq 3$). This result may be viewed as a generalization of a result in Ref.~\cite{Wills:2024wid} where a $2^{10}$-qudit magic state is converted to qubit $\ket{CCZ}$ states.

\begin{table}[h]
\caption{Asymptotically good triorthogonal classical codes.
\label{tab:agctc}}
\begin{tabular}{ccc}
 & \begin{tabular}[c]{@{}c@{}}$q = r^2$\end{tabular} & $\lim_{j \rightarrow \infty} \left(\frac{k_j}{n_j} + \frac{d_j}{n_j}\right)$ \\ \hline\hline
Wills et al.~\cite{Wills:2024wid}\footnote{A different notion of triorthogonality from~\cite[Definition A.3]{Wills:2024wid}.}  & $\geq 2^{10}$  & Not determined  \\ 
Nguyen~\cite{Nguyen:2024qwg}  & $\geq 2^8$  & $\geq \frac{1}{2}$\footnote{Deduced from~\cite[Theorem 3.6]{Nguyen:2024qwg}.}  \\
Present work  & $\geq 2^6$ & $\geq 1 - \frac{2r + 3}{3r(r-1)}$\footnote{Above the TVZ bound, i.e. $1 - 1/(r-1)$ for alphabet size $q = r^2$. }  \\ \hline\hline
\end{tabular}
\end{table}

\begin{theorem}[Informal version of the main results]

(1) Suppose the code $C_{j+1}:=C_\mathcal{L}(D_{j+1},G_{j+1})$ is the lifted code of $C_j := C_{\mathcal{L}}(D_j, G_j)$ in the sense of Def.~\ref{def:lifting}. Then, if $C_j$ is triorthogonal, so is $C_{j+1}$.

(2) For every $r=2^{m}$ with $m\geq 3$ and $q=r^{2}$, there exists a family of classical triorthogonal AG codes $\{C_j\}$ over alphabet size $q$ with parameters $[n_j,k_j,d_j]$ such that

\begin{align}
    \lim_{j \rightarrow \infty} \left(\frac{k_j}{n_j} + \frac{d_j}{n_j}\right) \geq 1 - \frac{2r + 3}{3r(r-1)} > 1 - \frac{1}{(r-1)}.
    \label{eq:tvz-summary}
\end{align}
Note that the latter inequality of Eq.~\eqref{eq:tvz-summary} is the TVZ bound for codes on $\mathbb{F}_{r^2}$. 
We summarize this result and compare with Refs.~\cite{Wills:2024wid,Nguyen:2024qwg} in Tab.~\ref{tab:agctc}. 

(3) These classical codes yield a family of good quantum codes over the alphabet $q = r^2$ for $r = 2^m \geq 8$ with parameters $[[\mathcal{N}_j, \mathcal{K}_j, \mathcal{D}_j]]_q$
\begin{align}
\label{eq:params2}
    \mathcal{N}_j &= r^j (r^2 - r) - \mathcal{K}_j\notag\\
    \mathcal{K}_j &= x_1 \, r^j \left( (r+1) \left\lfloor \frac{r-2}{3} \right\rfloor -2r\right) + x_2 \, v(r,j)\notag\\
    \mathcal{D}_j &\geq (1-x_1)\, r^j \, \left( (r+1) \left\lfloor \frac{r-2}{3} \right\rfloor -2r\right)\notag\\ &\phantom{xx}+ (1 - x_2)\,v(r,j) 
\end{align}
where 
\begin{align}
    v(r,j) := \begin{cases}
        4 r^{\frac{j+1}{2}}\, \text{ if j is odd}\notag\\
        2r^{\frac{j}{2}} \left(1 + r\right)\, \text{ otherwise},
    \end{cases}
\end{align}
the constants $x_1\in(0,1)$ and $x_2\in[0,1]$ are chosen such that 
$\mathcal{K}_j$ is an integer. 
We compare this result with Refs.~\cite{Wills:2024wid,Nguyen:2024qwg,GolowichGuruswami2024} in Tab.~\ref{tab:agqtc}. 

(4) Finally, given the ability to distill a state $\ket{CCZ}_{2^{2m}}$ (where $m \geq 3$), we can obtain the state $\ket{CCZ}_{2^n}$ for arbitrary $n$ using a constant-depth Clifford circuit with at most nine computational basis measurements, twelve single-qudit, and nine two-qudit Clifford gates.
\end{theorem}

\begin{table}[h]
\caption{Asymptotically good triorthogonal quantum codes: minimum qudit dimension $q_{min}$, bounds on encoding rate $\mathcal{R}$, and minimum code length $\mathcal{N}_{min}$. Previous works present various state and alphabet reduction procedures that allow for smaller $q$.
\label{tab:agqtc}}
\begin{tabular}{cccc}
 & $q_{min}$& $\mathcal{R}$ & $\mathcal{N}_{min}$\\ \hline\hline
Wills et al.~\cite{Wills:2024wid}  & $2^{10}$  & $\frac{5}{118} \leq \mathcal{R} \leq \frac{5}{114} $ & 932,090  \\ 
Nguyen~\cite{Nguyen:2024qwg}  & $2^8$  & $< \frac{1}{4}$\footnote{Deduced from \cite[Corollary]{Nguyen:2024qwg}}  & Not determined\\
Golowich et al.~\cite{GolowichGuruswami2024}  & $2^6$ & $> \frac{1}{100}$ & Not determined\\
Present work  & $2^6$ & $<  \frac{1}{2}  $\footnote{Obtained in the limit $x \rightarrow 1$, although at $x=1$ $\mathcal{D}_j=0$} & $\leq 55$ \\ 
\hline\hline
\end{tabular}
\end{table}



\section{Basic notions of Algebraic Geometry codes}
\label{sec:basic-notions}


Algebraic geometry (AG) codes are a class of error-correcting codes constructed from the framework of algebraic curves over finite fields. For a detailed review, we point readers toward Refs.~\cite{Chara2024,chara2025goodisodualagcodestowers,Stichtenoth2009}, and we will further adopt notation following Refs.~\cite{Wills:2024wid,Nguyen:2024qwg}. The basic idea is to exploit the structure of a curve — specifically, its rational points and the associated function field — to define codewords. In particular, by evaluating chosen functions at a set of rational points, one obtains linear codes whose parameters can surpass classical bounds. The theory draws on concepts from algebraic geometry to systematically relate the properties of the curve to the performance of the code. This interplay between geometry and coding theory has led to highly efficient classical codes, including ones that achieve the TVZ bound \cite{mana.19821090103}.
Many of the same ideas can be adapted to quantum error-correcting codes by using the CSS construction combined with classical AG codes to achieve high rates and large distances. 
This connection allows algebraic geometry to inform the design of quantum codes, enabling quantum analogues that inherit favorable properties from their classical counterparts. 

\subsection{Algebraic Geometry}

\begin{figure}
    \centering
    \includegraphics[width=\linewidth]{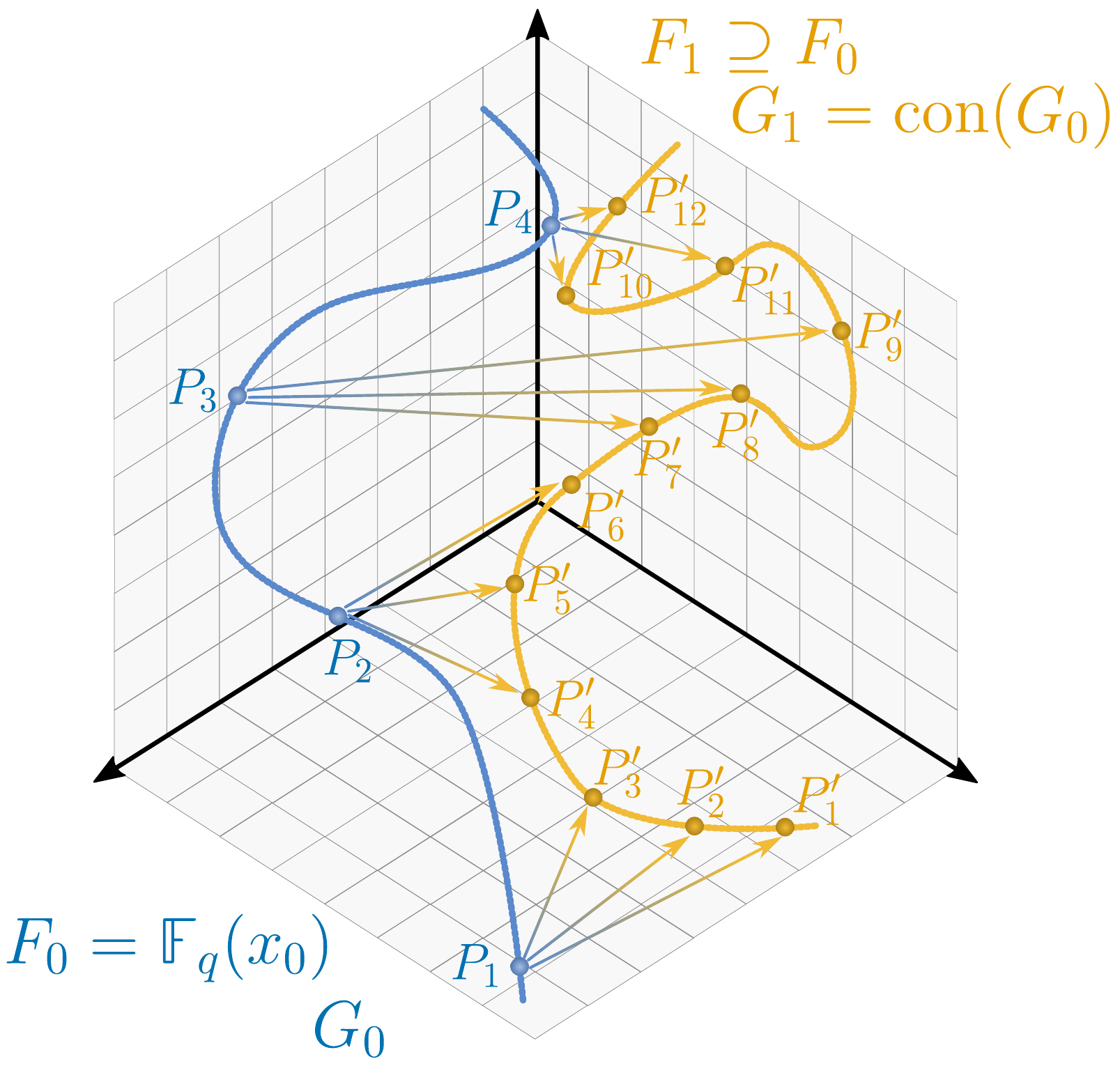}
    \caption{The blue curve corresponds to a function field $F_j$. Then $n$ pairwise distinct points $P_i$ correspond to rational places of $F_j$. The divisor $G_j$ specifies the Riemann-Roch space $\mathcal{L}(G_j)$. The points $P_i$ on the curves and the divisor $G_j$ specify an AG code of which the codewords are the vectors $(f(P_1), {...}, f(P_n))$ where $f \in \mathcal{L}(G_j)$. The yellow curve $F_{j+1}$ is an extension of $F_j$. An arrow $P_i \rightarrow P'_j$ indicates that $P'_j$ lies above $P_i$. This relation allows us to lift the code on $F_j$ into a code on $F_{j+1}$ while preserving triorthogonality. Then, we lift the Riemann-Roch space similarly by taking the conorm map of $G_{j+1} = \operatorname{Con}_{F_{j+1}/F_j}(G_j)$ to specify the lifted code. This builds sequences of triorthogonal codes by successive liftings.}  
    \label{fig:agcartoon}
\end{figure}

Before proceeding, let us define some preliminary notation that will be used throughout the paper. 
These concepts are well-known and basic to the field of algebraic geometry, and so describing our problem in these terms will allow us to borrow major results from that field to develop provably efficient AG codes for MSD purposes. 

\begin{definition}[Algebraic Function field] 
    An algebraic function field $F/K$ of one variable over $K$ is an extension field $F \supseteq K$ such that F is a finite algebraic extension of $K(x)$ for some element $x\in F$ that is transcendental over K. Associated with $F$ is a non-negative integer $g: = g(F)$, called the genus.
\end{definition}
We consider only the case where $K = \mathbb{F}_q$ is the full field of constants. 
From a function field we may build maximal ideals called ``places'' as well as divisors thereupon, which will serve as the central mathematical objects of AG code constructions. 
A place $P$ is uniquely related to a valuation ring $\mathcal{O}_P$ of the function field; such a ring satisfies: 
\begin{enumerate}
    \item $\mathbb{F}_q \subsetneqq \mathcal{O} \subsetneqq F$, and
    \item If $x\in F$, then $x\in \mathcal{O}$ or $x^{-1} \in \mathcal{O}$. 
\end{enumerate}

\begin{definition}[Places]
    A place $P$ is the maximal ideal of a valuation ring $\mathcal{O}$. 
    The set of all places of $F$ will be denoted $\mathbb{P}(F)$.
    We note the following attributes: 
    \begin{enumerate}
    
        \item A place $P$ has degree $\deg(P) = [F_P : \mathbb{F}_q]$ where $F_P := \mathcal{O}_{P}/P$, the residue field. The place $P$ is called rational if $\deg(P)=1$.  
        \item The valuation $v_P$ associated with a place $P$ is a function evaluated on $x\in F$ 
        as follows: 
        \begin{itemize}
        \item 
        If $x\neq0$ then, for any $t\in P$ such that $P=tO_P$, 
        there exist unique $n\in\mathbb{Z}$ and $u,u^{-1}\in\mathcal{O}_P$ such that 
        $x=t^n u$; 
        define $v_P(x):=n$. 
        \item 
        Define 
        $v_{P}(0):=\infty$ otherwise. 
        \end{itemize}
        \item For $x \in F$, a place $P$ is a zero of order $m$ if $v_{P}(x) = m > 0$; P is a pole of order $m$ if $v_{P}(x) = -m < 0$.
        \item A place $P'$ of a function field extension $F'$ lies above a place $P$ of $F$ if $\mathcal{O}_{P'}\cap F=\mathcal{O}_{P}$ (equivalently $P'\supseteq P$), and we denote $P'\mid P$. 
        The ramification of $P'$ over $P$ is $e(P'|P)\geq1$ such that $v_{P'}(x)=e(P'|P)v_P(x)$ for all $x\in F$. 
        \item A place $P$ splits completely in an extension $F'/F$ if $e(P'|P)=1$ for all  $P'\mid P$ (thereby ensuring $\sum_{P'|P}e(P'|P)=[F':F]$ places $P'$ lie above $P$). 
    \end{enumerate}
\end{definition}

From places we can construct certain collections called divisors, which will be the fundamental structures of our AG codes. 

\begin{definition}[Divisors]
    A divisor is a linear combination of places $M=\sum_P c_P P$, 
    with degree $\deg(M)=\sum_{P}c_{P}\deg(P)$ 
    and support $\mathrm{supp}(M)=\{P:c_P\neq0\}$. 
    Informally, we will say a place $P$ is in a divisor $M$ if $P \in \operatorname{supp}(M)$. 
    We say $M\geq0$ if $c_{P} \geq0$ for all $P$ in $M$. 
    We note the following kinds that will be used later: 
    \begin{enumerate}
    \item The canonical divisor associated with a differential $\omega$ is: 
    \begin{equation}(\omega) = \sum_{P\in\mathbb{P}(F)} v_P(\omega) P,
    \end{equation} 
    where $v_P(\omega):=v_P(z)$ given any decomposition $\omega = z\,dt$ with $z\in F$ and $tO_P=P$.\footnote{(For further reading, see Ref.~\cite[Chapter 5]{Stichtenoth2009} and Ref.~\cite[Chapter 12]{agcodes-advanced}).} 
    \item The conorm 
    of a divisor $M= \sum_{P} c_P P$ over $F$ 
    with respect to 
    a function field extension $F'/F$ is 
    \begin{equation}
        \mathrm{Con}_{F'/F}(M)=\sum_{P}\sum_{P'|P} e(P'|P) c_P P', 
        \label{eq:conorm}
    \end{equation}
    also a divisor. 
    \item The cotrace $\mathrm{Cotr}_{F'/F}(\omega):=\omega'$ of a differential $\omega$ with respect to $F'/F$ is the unique differential satisfying $\omega'(\alpha)=\omega(\mathrm{Tr}_{F'/F}(\alpha))$, where $\alpha$ is a map $\mathbb{P}(F')\to F'$ such that $\alpha_{P'}=\alpha_{Q'} \iff P'\cap F= Q'\cap F$. 
    Moreover, the cotrace divisor is 
    \begin{equation}
        \label{eq:cotr-con-diff}
        (\omega')=\mathrm{Con}_{F'/F}((\omega))+\mathrm{Diff}(F'/F),
    \end{equation}
    in terms of the well-known, ``different'' divisor $\mathrm{Diff}(F'/F)\geq0$. (For more details, see Ref.~\cite[Chapter 3]{Stichtenoth2009}.) 
    \end{enumerate}
\end{definition}

The well-known Riemann-Roch space of a divisor then constitutes a subspace of functions that we will use to realize the encoding in our AG codes. 
\begin{definition}[Riemann-Roch Space]
For a divisor $G$, the associated Riemann-Roch space denoted $\mathcal{L}(G)$ is the vector space 
\begin{align}
    \mathcal{L}(G) = \{x \in F : (x) + G \geq 0\} \cup \{0\}.
\end{align}
\end{definition} 
\noindent Following the Riemann-Roch theorem~\cite{Stichtenoth2009}, 
the Riemann-Roch space $\mathcal{L}(G)$ is characterized by the degree $\operatorname{deg}(G)$ and the genus $g(F)$.

\subsection{Codes from Algebraic Geometry}

Using the essential objects of algebraic geometry, we are ready to introduce generic AG codes. 
\begin{definition}[Algebraic geometry code]
    \label{def-ag-code}
    
    Let $F/\mathbb{F}_q$ be an algebraic function field.
    Let $D = P_1 + \cdots + P_n$ be a divisor consisting of pairwise distinct rational places,  and let $G =\sum_{P'}c_{P'}P'$ be another divisor of $F/\mathbb{F}_q$ disjoint from $D$.
    The AG code $C_{\mathcal{L}}(D,G)$ associated with the divisors $D$ and $G$ is the subspace 
\begin{align}
    C_{\mathcal{L}}(D,G) := \left\{\left(f(P_1), \ldots, f(P_n)\right) \mid f \in \mathcal{L}(G)\right\}
\end{align}
mapped by function field elements in the Riemann-Roch space $\mathcal{L}(G)$.

\end{definition}
Later, we will see that to construct suitable codes for MSD, it will be convenient to also characterize the dual of the AG code defined above. From Ref.~\cite{STICHTENOTH1988199}, the dual to the AG code $C_{\mathcal{L}}(D,G)$ is also an AG code, see the proposition below reproduced without proof.

\begin{definition}[Dual code]
    For a code $C \subseteq \mathbb{F}^n_q$, the dual $C^\perp$ is
    \begin{align}
        C^\perp := \left\{c' \in \mathbb{F}^n_q : \sum_{i=1}^n c_i c'_i = 0 \text{ for all } c \in C \right\}.
    \end{align}
\end{definition}

\begin{proposition}[Dual of $C_{\mathcal{L}}(D,G)$~\cite{STICHTENOTH1988199}]
    \label{prop:ag-dual}
    For an AG code $C_{\mathcal{L}}(D,G)$, there exists a differential $\eta \in \Omega_{F}$ with simple poles at each $P_i \in \operatorname{supp}(D)$ such that its residue at each $P_i$ is $\operatorname{res}_{P_i}(\eta) = 1$.
    Moreover, the dual of the code $C_{\mathcal{L}}(D,G)$ is
    \begin{align}
        C_{\mathcal{L}}(D,G)^{\perp} := C_{\mathcal{L}}(D, D - G + (\eta))
    \end{align}
    where $(\eta)$ is the canonical divisor associated with $\eta$. Finally, it follows from the Weak Approximation Theorem~\cite[Theorem 1.3.1]{Stichtenoth2009}, that there exists $z \in F$ such that the valuations $v_{P_i}(z) = 1$ for all $P_i \in \operatorname{supp}(D)$. Then, one choice of $\eta$ is the logarithmic differential
    \begin{align}
        \eta := \frac{dz}{z}.
    \end{align}
\end{proposition}
Clearly, $v_{P_i}(\eta) = -1$ and $\operatorname{res}_{P_i}(\eta) = 1$ for all $P_i \in \operatorname{supp}(D)$. Such a function $z \in F$ can be found  when the function field $F = \mathbb{F}_q(x)$. Our procedure only requires computing $\eta$ explicitly in the rational function field. 

The parameters $[n, k, d]$ of the AG codes $C:=C_{\mathcal{L}}(D,G)$ and $C^\perp=C_{\mathcal{L}}(D,G)^\perp$ can be calculated or bounded from the divisors $D$ and $G$ and the genus $g := g(F)$. 
 If $2g - 2 < \operatorname{deg}(G) < n$, then $k$ and $d$ of $C$ satisfy
        \begin{align}
        \label{eq:ag-code-parameters}
            k &= \operatorname{deg}(G) + 1 - g\notag\\
            d &\geq n - \operatorname{deg}(G).
        \end{align}
The distance of the dual code, $d^\perp$, satisfies
        \begin{align}
            \label{eq:ag-parameters-dual-code}
            d^\perp &\geq \operatorname{deg}(G) - (2g - 2).
        \end{align}
Details can be found in Ref.~\cite[Corollary 2.2.3 and Theorem 2.2.7]{Stichtenoth2009}.

\subsection{Triorthogonal codes}
Triorthogonal codes were first proposed in \cite{PhysRevA.86.052329} with the explicit purpose of making qubit magic state distillation more efficient. They are the most general stabilizer codes with transversal T gates, and important because they have low-overhead ($\gamma=\log_2(3)\approx1.6$) for universal qubit fault-tolerant quantum computing.

A triorthogonal code is a stabilizer code that comes from a special kind of matrix $\mathcal{G}$ with two orthogonality conditions: \emph{pairwise orthogonality} between any two rows and \emph{triple orthogonality} where the product of any three rows is 0. 
Subsequent research extended the triorthogonal code framework to qudits~\cite{Prakash:2024ghq,Sharma:2024zbn,Prakash:2025azi}.
For triorthogonality it is more practical formulate it via the star ($\star$) product of codes following Ref.~\cite{PhysRevLett.123.070507,Nguyen:2024qwg} where a code $C^{\star 2}$ is given by
    \begin{align}
        C^{\star2} := C\star C = \left\{c\star c' \mid c, c' \in C\right\}
    \end{align}
    and $c \star c'$ is the point-wise product of codewords $c$ and $c^\prime$, 
    \begin{align}
        c \star c' = (c_1c^\prime_1,\ldots, c_nc^\prime_n).
    \end{align}
    A code $C \subseteq \mathbb{F}^n_q$ satisfies the \emph{star-square property} if $C^{\star 2} \subseteq C^{\perp}$ and is equivalent to triple orthogonality. In addition to triple orthogonality, requiring the code to contain the constant word $(1,\ldots,1)$ ensures the code has pairwise orthogonality.  

\begin{definition}[Triorthogonal code]
    A linear code satisfying the Star-square property i.e. triple orthogonality i.e. triple orthogonality and  containing the constant word $(1, \ldots, 1)$ i.e. pairwise orthogonality is triorthogonal.
\end{definition}

In Lemma 2.3 of Ref.~\cite{Nguyen:2024qwg}, it is shown that a classical triorthogonal code yields a quantum code $\mathcal{Q}: \left[[\mathcal{N}, \mathcal{K}, \mathcal{D}\right]]_{\mathbb{F}_q}$ where the transversal $(CCZ)^{\otimes \mathcal{N}}$ gate implements the logical $\overline{CCZ}^{\otimes \mathcal{K}}$ gate. Following Refs.~\cite{Nguyen:2024qwg,Wills:2024wid}, it is possible to construct triorthogonal AG codes. From e.g. Ref.~\cite{Nguyen:2024qwg}, there is a sufficient condition for an AG code to be triorthogonal for power of two qudits, $q = 2^m$. For completeness, we will present this sufficient condition using Prop.~\ref{prop:relations-ag-codes} and Lemma~\ref{lemma:triorthogonal-ag-codes}.

\begin{proposition}[Some useful properties of AG codes]
    \label{prop:relations-ag-codes}
    Let $C := C_{\mathcal{L}}(D,G)$ and $C' := C_{\mathcal{L}}(D,G')$ be two AG codes with duals $C^{\perp} := C_{\mathcal{L}}(D,D - G + (\eta))$ and $(C')^{\perp} := C_{\mathcal{L}}(D,D - G' + (\eta'))$ respectively. Then

    \begin{enumerate}
        \item If $G \leq G'$ then $C \subseteq C^\prime$
        \item $C^{\star 2} \subseteq C_{\mathcal{L}}(D,2G)$.
    \end{enumerate}
\end{proposition}
\begin{proof}
    The first statement simply follows from the fact that $G \leq G^\prime$ implies $\mathcal{L}(G) \subseteq \mathcal{L}(G^\prime)$.  Then, every codeword $(f(P_1), {...}, f(P_n)) \in C$ must also be in $C^\prime$. Similarly, the second statement follows from the fact that for every pair $f,g \in \mathcal{L}(G)$, the product $fg$ is in $\mathcal{L}(2G)$.
\end{proof}

Using Proposition~\ref{prop:relations-ag-codes}, requiring $C_{\mathcal{L}}(D,2G) \subseteq C_{\mathcal{L}}(D,G)^{\perp}$ ensures $C^{\star 2} \subseteq C^{\perp}$. The lemma below will present sufficient conditions on the divisor $G$ to obtain a triorthogonal AG code. 

\begin{lemma}[Triorthogonal AG code]
    \label{lemma:triorthogonal-ag-codes}
    Let $C:= C_{\mathcal{L}}(D, G)$ with dual $C^{\perp} = C_{\mathcal{L}}(D, D - G + (\eta))$ where $\eta$ is as given in Prop.~\ref{prop:ag-dual}. 
    If 
    \begin{align}
        G &\geq 0 \; \text{  and  } \; 3G \leq D + (\eta),
    \end{align}
    then $(1, {...}, 1) \in C$ and $C^{\star 2} \subseteq C^{\perp}$.
\end{lemma}
\begin{proof}
  The first statement follows from the fact that for $G \geq 0$, the constant function $f = 1$ is in $\mathcal{L}(G)$ since it has no poles (and no zeros). The second statement arises from the fact that $2G \leq D - G + (\eta) \Leftrightarrow 3G \leq D + (\eta) $. Therefore, under this condition, we have $C_{\mathcal{L}}(D,2G) \subseteq C^{\perp}$ (Prop.~\ref{prop:relations-ag-codes}, first statement) and consequently $C^{\star 2} \subseteq C^{\perp}$(Prop.~\ref{prop:relations-ag-codes}, second statement).   
\end{proof}

\subsection{Tsfasman-Vladut-Zink bound}
One of our main results is that our triorthogonal codes are the first to exceed the TVZ bound~\cite{mana.19821090103}. Below, we review asymptotic bounds of classical codes and the TVZ bound. For an $[n, k, d]$ classical code $C$ over $\mathbb{F}_q$, the ratios:
\begin{align}
    R(C) := \frac{k}{n} \text{ and } \delta(C) := \frac{d}{n} 
\end{align}
are called encoding rate and relative distance respectively. The set 
\begin{align}
    V_q := \{ (\delta(C), R(C)) \mid C \text{ is a code over } \mathbb{F}_q  \} \subseteq [0,1]^2
\end{align}
contains all possible points $(\delta(C), R(C))$ where $C$ is a code over $\mathbb{F}_q$. The asymptotic behavior of codes over $\mathbb{F}_q$ can be understood by studying the set of limit points of $V_q$, denoted $U_q$. That is, a point $(\delta, R) \in [0,1]^2 $ is in $U_q$ if and only if there exists a sequence of codes $\{C_i\}_{i \geq 0}$ such that as $i \rightarrow \infty$, 
\begin{align}
    n_i &\rightarrow \infty, R(C_i) \rightarrow R \text{ and } \delta(C_i) \rightarrow \delta.
\end{align}

There exists a continuous function $\alpha_q$ such that $0 \leq R \leq \alpha_q(\delta)$ for all $\delta$. $\alpha_q$ satisfies $\alpha_q(0) = 1$ and  $\alpha_q(\delta) = 0$ for $1 - 1/q \leq \delta \leq 1$. The exact value of $\alpha_q(\delta)$ in the interval $0 < \delta < 1 - 1/q$ is unknown. However, there are bounds. Lower bounds are important because a non-zero lower bound guarantees the existence of good codes. One lower bound is the GV bound which random codes typically exceed. A significant achievement of algebraic geometry code construction is that it yields an improved bound, the so-called TVZ bound. For a square $q = r^2$, the TVZ bound is
\begin{align}
    \alpha_q(\delta) \geq 1 - \frac{1}{r-1} - \delta.
\end{align}
For $q \geq 49$ -- which is for all codes we consider -- the TVZ bound is larger than the GV bound for an increasing interval around $\delta\sim0.4$, showing that there are AG codes better than random codes. We will show that despite the extra algebraic constraint imposed by the triorthogonality requirement, our codes satisfy the TVZ bound. We call our codes triorthgonal TVZ codes.

\section{Classical Code Constructions}
\label{sec:classical-code}

Following Refs.~\cite{Nguyen:2024qwg,Wills:2024wid}, we construct a good family of triorthogonal codes $\{C_j\}_{j\rightarrow\infty}$ on qudits with dimension $q = r^2$ where $r\geq 8$ is a power of two. 
A ``good'' code means that both its rate $k_j/n_j$ and the relative distance $d_j/n_j$ do not vanish as $j\to\infty$.
Since we preserve the triorthogonal property, this family of good classical codes will yield quantum codes with transversal $CCZ$ gates for qudits of dimensions $q \geq 8^2$. 

Our construction is recursive. We consider a tower of function fields $\mathcal{W} = (F_0, F_1, {...})$, 
where $F_0 = \mathbb{F}_q(x_0)$ is the rational function field, 
$F_j$ is an extension of $F_{j-1}$ (for $j \geq 1$), and 
$\mathbb{F}_q$ is the full field of constants for each function field $F_j$. 
Each code $C_j$ is an AG code $C_j := C_{\mathcal{L}}(D_j, G_j)$, where each $D_j$ and $G_j$ are disjoint divisors in $F_j$ and the places in $D_j$ are rational. 

The simplicity of our construction is that we only need to construct the base code $C_0$ explicitly. This construction is straightforward since $F_0$ is the rational function field. Then, we present our lifting procedure (Sec.~\ref{sec:lifting}) to obtain the code $C_{j+1}$ using the $C_j$ code. Remarkably, we will show that this lifting procedure preserves triorthogonality (Thm.~\ref{thm:lifting}). Another convenient property of our construction is that the parameters of the $C_j$ codes can be derived or bounded using the parameters of $C_0$.
 
\subsection{Code lifting}
\label{sec:lifting}
Let the function field $F_{j+1}$ be an extension of the function field $F_j$. 
We assume that both function fields have the same field of constants $\mathbb{F}_q$. 
We can now use the conorm divisor defined in Eq.~\eqref{eq:conorm} with respect to these field extensions to lift codes recursively. 
\begin{definition}[Code lifting]
    \label{def:lifting}
    Let $C_j := C_{\mathcal{L}}(D_j, G_j)$ where $D_j = P_1 + \cdots + P_{n_j}$ includes only places that split 
    completely in the extension $F_{j+1}/F_j$. Set
    \begin{align}
        D_{j+1} &:= \operatorname{Con}_{F_{j+1}/F_j}(D_j),\notag\\ G_{j+1} &:= \operatorname{Con}_{F_{j+1}/F_j}(G_j).
    \end{align}
    Then $C_{j+1}:=C_\mathcal{L}(D_{j+1}, G_{j+1})$ is the lifting of $C_j$.
\end{definition}
Note that $C_{j+1}$ is a valid AG code (Sec.~\ref{sec:basic-notions}) because
\begin{enumerate}
    \item Since each place $P \in \operatorname{supp}(D_j)$ splits completely in the extension $F_{j+1}/F_j$, the places $P'\in \operatorname{supp}(D_{j+1})$ are rational places of $F_{j+1}$.
    \item Since $D_j$ and $G_j$ are disjoint, so are $D_{j+1}$ and $G_{j+1}$.
\end{enumerate}
Moreover, we can relate $C_{j+1}^\perp$ to $C_j^\perp$. 
\begin{lemma}[$C^\perp_{j+1}$ code]
    \label{lemma:dual-lifting} 
    Let $C_j := C_{\mathcal{L}}(D_j,G_j)$. Suppose $z \in F_j$ is as in Prop.~\ref{prop:ag-dual} so that $C^\perp_j = C_{\mathcal{L}}(D_j, D_j-G_j + (\eta_j))$ where $\eta_j = dz/z$. Regarding $z$ as an element of $F_{j+1}$, let us set $\eta_{j+1} = dz/z$, a differential in $F_{j+1}$. Then, $C_{j+1}^\perp = C_{\mathcal{L}}(D_{j+1}, D_{j+1} - G_{j+1} + (\eta_{j+1}))$. 
\end{lemma}
\begin{proof}
To prove this statement, we observe that any place $P'\in \operatorname{supp}(D_{j+1})$ lies above some place $P \in \operatorname{supp}(D_j)$. Then, $v_{P'}(z) = e(P'|P)v_{P}(z) = v_{P}(z) = 1$. Therefore, this choice of $\eta_{j+1}$ satisfies Prop.~\ref{prop:ag-dual}.
\end{proof}

Now, we have all the ingredients to establish one of our main results. More precisely, we show that our lifting operation preserves the multiplication property $C^{*2} \subseteq C^\perp$. 
\begin{theorem}[Lifting theorem]
    \label{thm:lifting}
    Suppose the code $C_j := C_\mathcal{L}(D_j, G_j)$ is as in Lemma~\ref{lemma:triorthogonal-ag-codes}, i.e. 
    \begin{align}
        G_j \geq 0 \text{ and } 3G_j \leq D_j + (\eta_j).
    \end{align}
    Then, so is the lifted code $C_{j+1}$.
\end{theorem}
\begin{proof}
    To prove this theorem, we observe that $G_{j+1} \geq 0$ is trivial since the Conorm of a positive divisor is also a positive divisor. The rest is to establish that $3G_{j+1} \leq D_{j+1} + (\eta_{j+1})$. 
    In fact, one can derive the formula 
\begin{align}
    &3G_{j+1} - D_{j+1} - (\eta_{j+1)})\notag\\ =& 
    \operatorname{Con}_{F_{j+1}/F_j}\left(3G_j - D_{j} - (\eta_j)\right) - \operatorname{Diff}(F_{j+1}/F_j)
\end{align}
by identifying our lifted differential with the cotrace of the lower differential, $\eta_{j+1}=\mathrm{Cotr}_{F'/F}(\eta_j)$, and using the cotrace divisor formula of Eq.~\eqref{eq:cotr-con-diff}. 
By assumption, we have $3G_j - D_{j} - (\eta_j) \leq 0$. Then, since the different divisor is positive, we arrive at $3G_{j+1} - D_{j+1} - (\eta_{j+1}) \leq 0$, concluding the proof. 
\end{proof}

In other words, our lifting procedure allows one to obtain new triorthogonal codes from existing ones.  Next, we relate the parameters of the $C_{j+1}$ code to those of $C_j$. 
\begin{lemma}[Parameters of $C_{j+1}$]
    \label{lemma:lifting-parameters}
    Let $C_j =C_{\mathcal{L}}(D_j, G_j)$ and $C_{j+1} = C_{\mathcal{L}}(D_{j+1}, G_{j+1})$
    , where $D_{j+1} = \operatorname{Con}_{F_{j+1}/F_j}(D_j)$ and $G_{j+1} = \operatorname{Con}_{F_{j+1}/F_j}(G_j)$. Then

    \begin{enumerate}
        \item Length $n_{j+1}$ satisfies
        \begin{align}
            n_{j+1} &= \operatorname{deg}(D_{j+1})\notag\\ &= [F_{j+1} : F_j]\operatorname{deg}(D_j)\notag\\
            &= [F_{j+1} : F_j] n_j.
        \end{align}
        \item Let $g_j := g(F_j)$ and $g_{j+1} := g(F_{j+1})$ be the genera 
        of the two function fields respectively. Suppose $2g_{j+1} - 2 < \operatorname{deg}(G_{j+1}) < n_{j+1}$. Then, the dimension $k_{j+1}$ is
        \begin{align}
            k_{j+1} &= \operatorname{deg}(G_{j+1}) + 1 - g_{j+1}\notag\\
            &= [F_{j+1}:F_j] \operatorname{deg}(G_j) + 1 - g_{j+1}.
        \end{align}
        \item Distance $d_{j+1}$ satisfies
        \begin{align}
            d_{j+1} &\geq n_{j+1} - \operatorname{deg}(G_{j+1})\notag\\
            &= [F_{j+1}:F_j]\left(n_j - \operatorname{deg}(G_j)\right).
        \end{align}
    \end{enumerate}
\end{lemma}
\begin{proof}
   These statements follow from~\ref{eq:ag-code-parameters} and from the fact that for a divisor $A \in \operatorname{Div}(F_j)$, $\operatorname{deg}(\operatorname{Con}_{F_{j+1}/F_j} (A)) = [F_{j+1}: F_j] \operatorname{deg} (A)$.
 
\end{proof}

To end this discussion, we point that out there exists an alternative lifting method (provided in Appendix~\ref{app:alt-lift}) similar to that used in Ref.~\cite{Chara2024}. 
While such a lifting improves the dimension of the code, this improvement will be sub-leading for the family of codes we will build. 
Therefore, we will use the simpler lifting to analyze the asymptotic behavior of our codes. 

\subsection{The Family of Codes}

The family of codes $\{C_j\}_{j \rightarrow \infty}$ is constructed from a tower  $\mathcal{W}$ of function fields over $\mathbb{F}_{q}$ with $q = r^2$. Using this tower, we show that for $r \geq 8$, there exists a good family of codes $\{C_j\}_{j \rightarrow \infty}$ in which each code $C_j$ is triorthogonal. The tower we use is described in Refs.~\cite{agcodes-advanced,Stichtenoth2009}. 

\begin{definition}
    \label{def:tower}
    Let $\mathcal{W} = (F_0, F_1, F_2, \ldots)$ be the tower of function field over $\mathbb{F}_{q}$ with $q = r^2$ where 
    \begin{enumerate}
        \item $F_0 := \mathbb{F}_q(x_0)$ is the rational function field.
        \item $F_{j+1} := F_{j}(x_{j+1})$ and each $x_j$ and $x_{j+1}$ satisfy the relation
        \begin{align}
            x^{r}_{j+1} + x_{j+1} = \frac{x^r_j}{x^{r-1}_j + 1}.
        \end{align}
    \end{enumerate}
\end{definition}
Many useful properties of this tower are already known~\cite[Chapter 7]{agcodes-advanced}. 
\begin{proposition}[Properties of tower]
    Consider the tower $\mathcal{W}$ in \Cref{def:tower}.
    \begin{enumerate}
        \item For each $j \geq 0$, the extension $F_{j}/F_0$ has degree $[F_{j} : F_0] = r^j$.
        \item The rational places in $F_0$ form the set
        \begin{align}
            \mathcal{Z} := \left\{P_{\alpha} : \alpha \in \mathbb{F}_{r^2} \backslash \mathbb{F}_{r} \right\}, 
        \end{align}
        and each splits completely in all extensions $F_{j}/F_0$. This set has cardinality $|\mathcal{Z}| = r(r-1)$.
        \item The ramification locus of the tower is
        \begin{align}
            \mathcal{V} := &\left\{P_{\alpha} : \alpha \in \mathbb{F}_r \right\} \cup \{P_{\infty}\}.
        \end{align}
        This set has cardinality $|\mathcal{V}| = r + 1$.
        \item The genus $g_j := g(F_j)$ of the function field $F_j$ is
        \begin{align}
            \label{eq:genus}
            g_j &= \begin{cases}
                \left(r^{\frac{j+1}{2}} - 1\right)^2 &\text{if \,$j$\, is odd}\\
                \left(r^\frac{j}{2} - 1\right)\left(r^{\frac{j+2}{2}} - 1\right) &\text{otherwise}
            \end{cases}
        \end{align}
        Often, we will use the expression
        \begin{align}
            \label{eq:1-genus}
            1 - g_j &= \begin{cases}
                -r^{j+1} + 2 r^{\frac{j+1}{2}} &\text{if \,$j$\, is odd}\\
                 -r^{j+1} + r^{\frac{j}{2}} (1 + r) &\text{otherwise},
            \end{cases}
        \end{align}
        which can be derived from \eq{genus}.
    \end{enumerate}
\end{proposition}


Now, we are ready to present our family of triorthogonal codes. Our construction is recursive in a simple form: we first use the rational function field $F_0 := \mathbb{F}_q(x_0)$ to construct the base code $C_0$, and then obtain each subsequent code $C_{j+1}$ from $C_j$ via the triorthogonality-preserving lifting operation. 
Moreover, the parameters of all codes in the sequence can be computed or bounded from those of $C_0$. We summarize this result in the theorem below. 

\begin{theorem}[Good family of triorthogonal codes]
    For $r = 2^m$ with $m \geq 3$, there exists a good family of triorthogonal codes $\left\{C_j\right\}_{j \rightarrow \infty}$ over the alphabet $\mathbb{F}_{r^2}$ with parameters $\left[n_j, k_j, d_j\right]$ given by
    \begin{align}
        n_j &= r^{j}(r^2-r)\notag\\
        k_j &= r^{j} (r+1) \left\lfloor \frac{r-2}{3} \right\rfloor + 1 - g_j\notag\\
        d_j &\geq  \frac{2}{3} r^j (r^2 - r + 2).
        \label{eq:good-triorth-codes}
    \end{align}

    One explicit family realizing these parameters is defined recursively by
    \begin{align}
    C_{0} &= C_{\mathcal{L}}(D_0, G_0),\notag\\
    C_{j+1} &= C_{\mathcal{L}}(\operatorname{Con}_{F_{j+1}/F_j}(D_j), \operatorname{Con}_{F_{j+1}/F_j}(G_j)),
    \end{align}
    where 
    \begin{align}
    D_0 = \sum\limits_{P_{\alpha} \in \mathcal{Z}} P_{\alpha}, \;\;\; G_0 = \left\lfloor \frac{r-2}{3} \right\rfloor \sum_{Q_{\alpha} \in \mathcal{V}} Q_{\alpha}.
\end{align}
and $F_j \in \mathcal{W}$, the tower of Def.~\ref{def:tower}. 
\end{theorem}
\begin{proof}
We prove this theorem by construction. By Thm.~\ref{thm:lifting}, it is sufficient to explicitly construct the base code $C_0 := C_\mathcal{L}(D_0, G_0)$. We use the set $\mathcal{Z}$ to construct the divisor $D_0$ and the set $\mathcal{V}$ to construct $G_0$. Set
\begin{align}
    D_0 &:= \sum_{P\in \mathcal{Z}} P && 
    G_0 := \sum_{Q \in \mathcal{V}} a_Q Q 
\end{align}
where 
the coefficients $a_Q \in \mathbb{Z}$ will be chosen so  that $C_0$ is triorthogonal.

From Ref.~\cite[Prop.~8.1.2]{Stichtenoth2009}, a differential $\eta_0$ as in Prop.~\ref{prop:ag-dual} is given by
\begin{align}
    \eta_0 := \frac{dt_0}{t_0}, \quad t_0 := \prod\limits_{\alpha \in \mathcal{Z}} (x_0 + \alpha).
\end{align}
A direct calculation gives the canonical divisor (outlined in Appendix~\ref{app:divisor}) 
\begin{align}
    (\eta_0) = - D_0 +  (r - 2) \sum\limits_{Q \in \mathcal{V}} Q.
\end{align}
Therefore, we can ensure $C_0$ is triorthogonal by setting $a_Q := \lfloor (r-2)/3\rfloor$ for $r \geq 8$. Hence, the base code is $C_0 := C_{\mathcal{L}}(D_0, G_0)$
with
\begin{align}
    D_0 &= \sum\limits_{P \in \mathcal{Z}} P,\notag\\
    G_0 &= \left\lfloor \frac{r-2}{3} \right\rfloor \sum_{Q \in \mathcal{V}} Q.
\end{align}
The parameters of $C_0$ are
\begin{align}
    n_0 &= r(r-1),\notag\\
    k_0 &= \operatorname{deg}(G_0) + 1 = (r + 1)  \left\lfloor \frac{r-2}{3} \right\rfloor + 1,\notag\\
    d_0 &= r(r-1) - (r + 1) \left\lfloor \frac{r-2}{3} \right\rfloor.
    \label{eq:base-code-params}
\end{align}
For example, for $r = 8$, $C_0$ is a triorthogonal code with parameters $[56, 19, 38]_{8^2}$.

We now define the family by successive lifting. Since $P \in \mathcal{Z}$ and $Q \in \mathcal{V}$, we can ensure that for each $j$:
\begin{enumerate}
    \item Every place in $\operatorname{supp}(D_j)$ splits completely in the extension $F_{j+1}/F_j$. 
    \item The divisors $D_{j+1}$ and $G_{j+1}$ are disjoint.
\end{enumerate}
Therefore, each $C_j$ can be lifted to $C_{j+1}$ by Def.~\ref{def:lifting}:
\begin{align}
    C_{0} &= C_{\mathcal{L}}(D_0, G_0),\notag\\
    C_{j+1} &= C_{\mathcal{L}}(\operatorname{Con}_{F_{j+1}/F_j}(D_j), \operatorname{Con}_{F_{j+1}/F_j}(G_j)),
\end{align}
with $F_j \in \mathcal{W}$. 
By Thm.~\ref{thm:lifting},  each $C_j$ is triorthogonal. 
Lemma~\ref{lemma:lifting-parameters} in combination with Eq.~\eqref{eq:base-code-params} and
\begin{align}
    \operatorname{deg}&G_j>  2g_j - 2 \text{ for } r \geq 8.
\end{align}
make clear that Eq.~\eqref{eq:good-triorth-codes} holds. 

Now, it is easy to show that this family is asymptotically good, i.e. the ratios $d_j/n_j$ and $k_j/n_j$ do not vanish. 
\begin{align}
    \label{eq:dist-limit}
    \lim_{j\rightarrow\infty} \frac{d_j}{n_j} &\geq \lim_{j\rightarrow\infty} \frac{r^{j}(r^2 - r + 2)}{r^{j}(r^2-r)}\notag\\
    &= \frac{2}{3} + \frac{4}{3r(r-1)}
\end{align}
For the dimension, we can use Eq~(\ref{eq:1-genus}) to get
\begin{align}
    \label{eq:rate-limit}
    \lim_{j\rightarrow\infty} \frac{k_j}{n_j} 
    &= \lim_{j\rightarrow\infty} \frac{r^{j}\left[(r+1)\left\lfloor \frac{r-2}{3} \right\rfloor - 1\right]}{r^{j}r(r-1)}\notag\\
    &\geq \frac{1}{3} - \frac{2r+7}{3r(r-1)} > 0 \text{ for } r \geq 8.
\end{align}
\end{proof}

These limits show that our codes have very good encoding rates for their relative distance. More precisely, these codes explicitly satisfy the TVZ bound for $r\geq8$:
\begin{align}
    \lim_{j\rightarrow\infty} \frac{k_j}{n_j} + \frac{d_j}{n_j} \geq 1 - \frac{2r + 3}{3r(r-1)} > 1 - \frac{1}{r-1}.
\end{align}
Thus, there exists a good family of triorthogonal codes $\{C_j\}$ exceeding the TVZ bound on finite Fields $\mathbb{F}_{r^2}$ for $r \geq 8$. Moreover, it can be verified by direct computations that these codes are also above the GV bound, thereby scaling better than random codes in all regimes. As we will see, such large encoding rates are beneficial when converting to quantum codes. Indeed, the dimension of the quantum code will be $\mathcal{K} \leq k$ and the length $\mathcal{N} = n - \mathcal{K}$. Hence, a large encoding rate in the classical code yields the possibility of quantum codes with correspondingly large encoding rates. In fact, a large encoding rate in the classical code will result in a large flexibility to either increase the distance of the quantum code or increase its dimension instead, as desired. 

\section{Quantum Codes}
\label{sec:quantum-codes}

In this section, we outline how to obtain qudit CSS codes from the previously constructed family of classical triorthogonal codes that support a transversal $CCZ$ gate. We first briefly review the basic elements of quantum computation on qudits of dimension $q = 2^m$, $m \geq 1$, that are needed for constructing error-correcting codes. 

The single-qudit Hilbert space is simply 
\begin{align}
    \mathcal{H}_q := \operatorname{span}\left\{ \ket{\alpha} \;|\; \alpha \in \mathbb{F}_q\right\}.
\end{align}
The qudit Pauli gates $X^\beta$ and $Z^{\beta}$ ($\beta \in \mathbb{F}_q)$ act as
\begin{align}
    X^{\beta} \ket{\alpha} &= \ket{\alpha + \beta},\notag\\
    Z^{\beta}\ket{\alpha} &= (-1)^{\operatorname{Tr}_{\mathbb{F}_q / \mathbb{F}_2}(\alpha\beta)} \ket{\alpha}.
    \label{eq:qudit-paulis}
\end{align}
Here, we have used the field trace operator $\operatorname{Tr}_{\mathbb{F}_{r^n}/\mathbb{F}_{r}}(\cdot):  \mathbb{F}_{r^n} \to \mathbb{F}_r$ defined by
\begin{align}
    \operatorname{Tr}_{\mathbb{F}_{r^n}/\mathbb{F}_{r}}(x) &:= x + x^r + x^{r^2} + \cdots + x^{r^{n-1}}
\end{align} 
to describe these phase gates. 
These gates generate the qudit Pauli group $\mathcal{P}\subset U(q)$, whose elements take the form $i^a X^{\alpha} Z^\beta$ where $a \in \{0,1,2,3\}$ and $\alpha, \beta \in \mathbb{F}_q$. 
The normalizer of the qudit Pauli group is the qudit Clifford group $\mathcal{C}$:
\begin{align}
    \mathcal{C} = \left\{U\in U(q) \mid U P U^\dagger \in \mathcal{P} \; \text{for all} \; P \in \mathcal{P} \right\}.
\end{align}

This qudit Clifford group and a non-Clifford gate element form a universal gate set for qudit-based quantum computation. For our construction of quantum codes, we consider the $CCZ$ gate, a non-Clifford gate defined by
\begin{align}
    CCZ \ket{x}\ket{y}\ket{z} = (-1)^{\operatorname{Tr}_{\mathbb{F}_{q} / \mathbb{F}_2}(xyz)} \ket{x}\ket{y}\ket{z}.
\end{align}
This gate lies in the third level of the Clifford hierarchy, implying a $\ket{CCZ}$ state given in $\mathcal{H}_q^{\otimes 3}$ by 
\begin{align}
    \ket{CCZ} &:= CCZ\ket{+}\ket{+}\ket{+}\notag\\ &=\frac{1}{q^{3/2}} \sum_{x,y, z \in \mathbb{F}_q} (-1)^{\operatorname{Tr}_{\mathbb{F}_q / \mathbb{F}_2}(x y z)} \ket{x}\ket{y}\ket{z},
\end{align}
where $\ket{+}$ is the uniform superposition of basis states. The state $\ket{CCZ}$ can be injected using a Clifford circuit via gate teleportation \cite{Knill:2004ctr}. So, the task of MSD in this paper is to distill this state. The next sections will describe the quantum codes that we will build for this procedure.

We can then discuss how to obtain a quantum code $\mathcal{Q}$ with parameters $[[\mathcal{N},\mathcal{K},\mathcal{D}]]_{r^2}$ from a triorthogonal code $C : [n,k,d]_{r^2}$ following a standard procedure~\cite{Nguyen:2024qwg,PhysRevLett.123.070507}. 
In short, the $k\times n$ generator matrix of a triorthogonal code $C$ can be put in the convenient form
\begin{align}
    \label{eq:gen-matrix}
    \mathcal{G}(C) = \begin{pmatrix}
        \mathbb{1}_{\mathcal{K}} & H_1\\
        0 & H_0
    \end{pmatrix},
\end{align}
where $\mathcal{K} \leq k$. 
Because $C$ is triorthogonal, the 
row vector subspaces, $\mathcal{H}_0 := \operatorname{rowspace}(H_0)$ and $\mathcal{H}_1 := \operatorname{rowspace}(H_1)$, 
are orthogonal 
(see, e.g., Ref.~\cite{Nguyen:2024qwg}). 
Let $\mathcal{H} := \operatorname{rowspace}(H_0, H_1)$, so $\mathcal{H}_0 \subset \mathcal{H}$. 
From the qudit Pauli gates defined in Eq.~\eqref{eq:qudit-paulis} combined with these spaces, we can define the quantum CSS code 
\begin{align}
    \mathcal{Q} = \operatorname{CSS}\left(X, \mathcal{H}_0; Z, \mathcal{H}^\perp\right),
\end{align}
which has 
a transversal $CCZ$~\cite[Lemma 2.3]{Nguyen:2024qwg}.

The parameters $\mathcal{N}, \mathcal{K}$, and $\mathcal{D}$ of this quantum code can be estimated from those of the classical code. Clearly, the length is $\mathcal{N} = n - \mathcal{K}$, and $\mathcal{K}$ is the dimension. As a CSS code, its distance is given by  $\mathcal{D} = \operatorname{min}(\mathcal{D}_X, \mathcal{D}_Z)$
where 
\begin{align}
    \mathcal{D}_X &= \operatorname{min}_{f \in \mathcal{H}\backslash\mathcal{H}_0} |f| \geq \operatorname{dist}(\mathcal{H}) \geq d - \mathcal{K},\notag\\
    \mathcal{D}_Z &= \operatorname{min}_{f \in \mathcal{H}^\perp_0\backslash \mathcal{H}^\perp} |f| \geq \operatorname{dist}(\mathcal{H}^\perp_0)
\end{align}
and $d$ is the distance of the classical code. 
Assuming $2g - 2 < \operatorname{deg}G < n$, Eq.~\ref{eq:ag-code-parameters} yields
\begin{align}
    \mathcal{D}_X \geq n - \operatorname{deg}G - \mathcal{K}.
\end{align}
Along with this, we can lower bound $\mathcal{D}_Z$. As explained in Ref.~\cite{Nguyen:2024qwg}, the code $\mathcal{H}_0$ can be viewed as an AG code
\begin{align}
    \mathcal{H}_0 = C_{\mathcal{L}}(D', G').
\end{align}
Here the divisor $D' = \sum_{i=\mathcal{K}+1}^{n} P_i$ corresponds to puncturing $\mathcal{K}$ positions, and
$G' = G - \sum_{i=1}^{\mathcal{K}} P_i$.
Equivalently, $\mathcal{H}_0$ is obtained by restricting the Riemann–Roch space $\mathcal{L}(G)$ to functions that vanish at all punctured positions, so that $\mathcal{L}(G') \subset \mathcal{L}(G)$.
Then, using Eq.~\eqref{eq:ag-parameters-dual-code}, we obtain
\begin{align}
    \mathcal{D}_Z \geq\operatorname{dist}(\mathcal{H}^\perp_0) &\geq \operatorname{deg}(G') - (2g - 2)\notag\\
    &= \operatorname{deg}(G) - \mathcal{K} - (2g - 2).
\end{align}
In our case, we will show $\mathcal{D}_Z <\mathcal{D}_X$, implying that the distance of each of our quantum codes satisfy
\begin{align}
    \label{eq:quantum-dist-bound1}
    \mathcal{D} \geq \operatorname{deg}(G) - \mathcal{K} - (2g - 2).
\end{align}


In particular, we may derive a family of 
quantum codes $\mathcal{Q}_j$ from the triorthogonal codes $C_j$ constructed in Sec.~\ref{sec:classical-code}. 
We summarize our quantum codes below:
\begin{theorem}[Good quantum codes with a transversal CCZ gate]
    \label{thm:quantum-codes}
    There exists a good family of quantum codes $\{\mathcal{Q}_j\}$ with transversal CCZ gates on qudits with dimension $q = r^2=2^{2m}$ for $m \geq 3$. Each code $\mathcal{Q}_j$ has parameters

\begin{align}
    \mathcal{N}_j &= r^j (r^2 - r) - \mathcal{K}_j\notag\\
    \mathcal{K}_j &= x_1 \, r^j \left( (r+1) \left\lfloor \frac{r-2}{3} \right\rfloor -2r\right) + x_2 \, v(r,j)\notag\\
    \mathcal{D}_j &\geq (1-x_1)\, r^j \, \left( (r+1) \left\lfloor \frac{r-2}{3} \right\rfloor -2r\right)\notag\\ &\phantom{xx}+ (1 - x_2)\,v(r,j) 
\end{align}
where 
\begin{align}
    v(r,j) := \begin{cases}
        4 r^{\frac{j+1}{2}}\, \text{ if j is odd}\notag\\
        2r^{\frac{j}{2}} \left(1 + r\right)\, \text{ otherwise}
    \end{cases}
\end{align}
and the constants $x_1\in(0,1)$ and $x_2\in[0,1]$ are chosen such that $\mathcal{K}_j$ is an integer.
\end{theorem}
 
\begin{proof}
   We show the proof in App.~\ref{app:quantum-codes}. 
\end{proof}


\section{Explicit Code Constructions}
\label{sec:explicit-codes}

Having demonstrated the existences of new families of asymptotically or nearly asymptotically good codes, in this section we undertake to present explicit results for the properties of some $\mathcal{Q}_j=[[\mathcal{N}_j,\mathcal{K}_j,
\mathcal{D}_j]]_q$ for low $q$ and $j$.  In particular, we are interested in understanding what overheads $\gamma$ and encoding rates $\mathcal{R}=\mathcal{K}_j/\mathcal{N}_j$ are achievable. For reference in this discussion, comparisons will be made to the canonical code used in qubit MSD of $\mathcal{Q}_{BK}=[[15,1,3]]_2$ with $\gamma\approx2.46$ and $\mathcal{R}\approx0.07$~\cite{PhysRevA.71.022316}.

While $\mathcal{R}$ is straightforward once the code is constructed, computing the minimum-weight logical operator---and therefore the distance--for an arbitrary stabilizer code is NP-complete.  For the large $q$ codes that we will be considering, we thus have to suffice with bounds for now.  For lower bounds $\mathcal{D}_j^{min}$, we use Eqs.~(\ref{eq:params2}) and~(\ref{eq:quantum-dist-bound2}).  To estimate upper bounds, for $Q_0$ we can leverage the \textsc{qLDPC} code available at~\cite{perlin2023qldpc} with a workstation. We leave the study of larger codes to future work. 

The lowest $q$ for which we can construct asymptotically good codes is the family of codes with $q=2^6$. Their investigation is further motivated by the experimental demonstration of $q=64$ state preparation, suggesting MSD with these codes is nearer to realization than larger $q$. For $j\leq4$, the bounds on $\gamma$ are presented as a function of $\mathcal{R}$ for all possible codes in our family with $q=2^6$ in Fig.~\ref{fig:bounds8}. 
The first point to observe is that $\gamma^{max}\rightarrow 0$ as $j\rightarrow\infty$ (as anticipated for an asymptotically good family). Moreover, all codes outperform $\mathcal{Q}_{BK}$.  In particular, the best performing code in $Q_0$ is $[[42,14,6]]_{64}$, which has $\mathcal{R}=\frac{1}{3}$ and $\gamma^{max}\approx0.613$. The codes within each $\mathcal{Q}_j$ with the lowest $\gamma^{max}$ are tabulated in Table~\ref{tab:bestcodes}.

\begin{figure}
    \centering
    \includegraphics[width=1\linewidth]{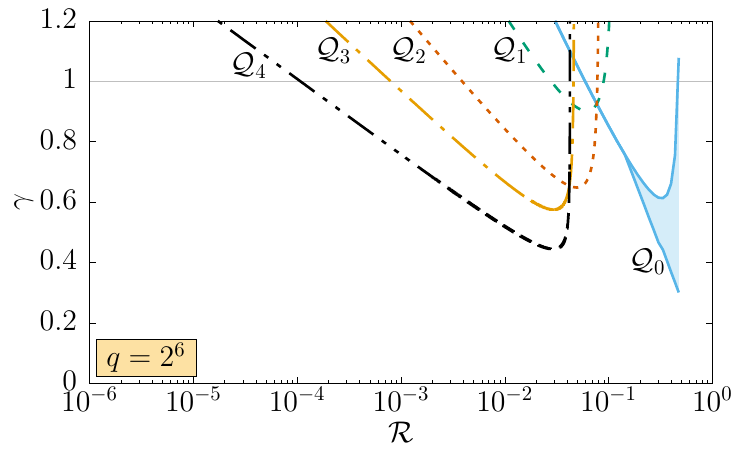}
    \caption{Bounds on $\gamma$ vs. $\mathcal{R}$ for quantum codes $\mathcal{Q}_j$ at different levels of the lifting for $q=2^6=64$.  For $\mathcal{Q}_0$, we present two-sided bounds, while for the rest we have only upper bounds.}
    \label{fig:bounds8}
\end{figure}

Increasing the qudit dimension to $q=2^8$; $\mathcal{N}_j$, $\mathcal{K}_j$, and $\mathcal{D}^{min}_j$ all grow by factors of 3-4.  While the best code in $\mathcal{Q}_0$, $[[188,52,18]]$,  demonstrates a smaller overhead at $\gamma^{max}=0.445$ than the best code in $\mathcal{Q}_4$ of $d=2^6$, it also comes at a cost of less efficient encoding rate $\mathcal{R}$ (See Table~\ref{tab:bestcodes}).  When $j>0$, all $\mathcal{Q}_j$ have lower overhead $\gamma^{max}$ than any code we have constructed with $d=2^6$ and greater encoding rates by a factor of 3.  

\begin{figure}
    \centering
    \includegraphics[width=1\linewidth]{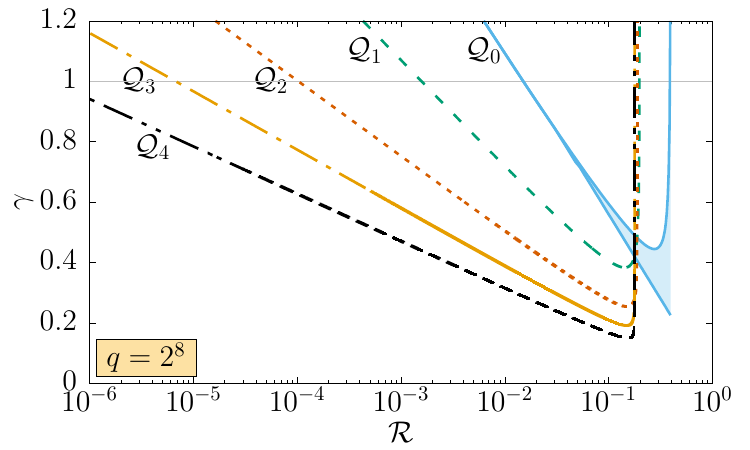}
    \caption{Bounds on $\gamma$ vs. $\mathcal{R}$ for quantum codes $\mathcal{Q}_j$ at different levels of the lifting for $q=2^8=256$.  For $\mathcal{Q}_0$, we present two-sided bounds, while for the rest we have only upper bounds.}
    \label{fig:bounds16}
\end{figure}

The final qudit dimension for which we explicitly construct codes is $q=2^{10}$, so that we may compare with the family of asymptotically good codes identified in Ref.~\cite{Wills:2024wid}.  In that work, the codes had $\mathcal{R}\leq \frac{5}{114}\approx0.04$ and $\mathcal{N}\geq 932,090$ were identified.  In contrast, our families have $\mathcal{R}\geq 0.319$ and $\gamma^{max}\leq0.236$ for the best codes in all $\mathcal{Q}_j$.  Further, it is only in $\mathcal{Q}_2$ where our codes begin to exceed $\mathcal{N}=932,090$.  Instead, our best code in $\mathcal{Q}_0$ requires only $\mathcal{N}=708$.

\begin{figure}
    \centering
    \includegraphics[width=1\linewidth]{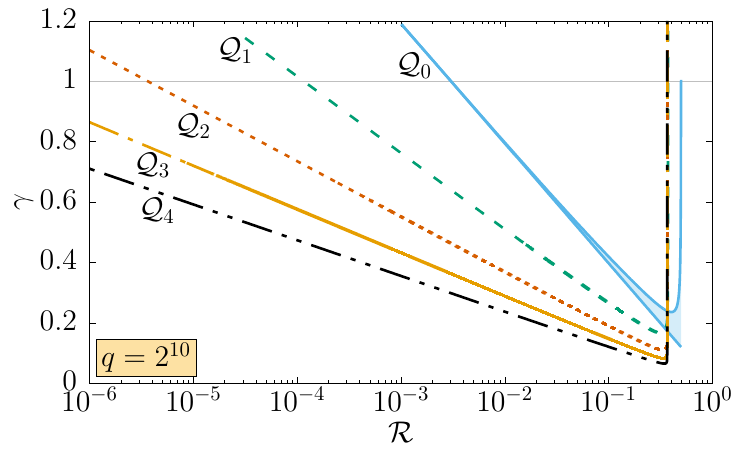}
    \caption{Bounds on $\gamma$ vs. $\mathcal{R}$ for quantum codes $Q_j$ at different levels of the lifting for $q=2^{10}=1024$.  For $\mathcal{Q}_0$, we present two-sided bounds, while for the rest we have only upper bounds.}
    \label{fig:bounds32}
\end{figure}

In all such cases considered here, we find low overhead costs to TVZ codes that we have constructed explicitly per the methods of Secs.~\ref{sec:classical-code} and \ref{sec:quantum-codes}. 
Typically, when constructing codes, one aims to guarantee constant overhead by choosing a sufficiently large code from a family, so that only a single round of MSD is required.
However for practical implementation this task may be too ambitious and less cost-efficient overall. 
Indeed, we have constructed codes that are sufficiently small such that they would contribute a low cost even with overhead over multiple rounds of MSD.

One outstanding question is the error threshold for distillation of $CCZ$ from these codes. Prior work on qutrits and ququints found low-overhead codes with thresholds to depolarizing noise as large as $\epsilon\approx 0.1$ exist~\cite{Prakash:2024ghq,Prakash:2025azi,Saha:2025frb}. Given the higher dimension of the qudits here, the many error channels likely result in lower thresholds, which would degrade the benefits from the codes proposed here.  Therefore error thresholds should be investigated while varying $\mathcal{R}$.   Another question would be to quantify the codes' maximum stabilizer weight checks, as prior work has shown them to be potentially non-constant in $\mathcal{N}$.  

\begin{table}[h]
\caption{Summary of the codes $Q_j$ for three values of $r$ with the upper bound $\gamma^{max}$ on their overhead.  Reported are: the number of physical qudits used $N_j$, the number of logical qudits encoded $K_j$, and a lower bound on the code distance $D_j^{min}$, and the encoding rate $\mathcal{R}$.
\label{tab:bestcodes}}
\centering
\begin{tabular}{c| c| c c c}
$r$ & $j$ & $[[\mathcal{N}_j,\mathcal{K}_j,
\mathcal{D}^{min}_j]]$ &$\mathcal{R}$& $\gamma^{max}$ \\
\hline\hline
\multirow{5}{*}{8}
 & 0 & [[42,14,6]] & 0.333 & 0.613 \\
 & 1 & [[422,26,22]] & 0.062 & 0.902 \\
 & 2 & [[3416,168,104]] & 0.049 & 0.649 \\
 & 3 & [[27851,821,459]] & 0.030 & 0.575 \\
 & 4 & [[222850,6526,2818]] & 0.029 & 0.444 \\
\hline
\multirow{5}{*}{16}
 & 0 & [[188,52,18]] & 0.277 & 0.445 \\
 & 1 & [[3361,479,161]] & 0.143 & 0.383 \\
 & 2 & [[53447,7993,1767]] & 0.150 & 0.254 \\
 & 3 & [[855756,127284,21196]] & 0.149 & 0.191 \\
 & 4 & [[13632492,2096148,271852]] & 0.154 & 0.150 \\
\hline
\multirow{5}{*}{32}
 & 0 & [[708,284,48]] & 0.401 & 0.236 \\
 & 1 & [[24072,7672,968]] & 0.319 & 0.166 \\
 & 2 & [[762355,253453,21043]] & 0.333 & 0.111 \\
 & 3 & [[24292183,8213673,506711]] & 0.338 & 0.083 \\
 & 4 & [[774113521,266073871,12914929]] & 0.344 & 0.065 \\
\hline\hline
\end{tabular}
\end{table}

\section{State Reduction}
\label{sec:msd}

In Thm.~\ref{thm:quantum-codes}, we have constructed good quantum codes allowing constant-overhead MSD on qudits of dimension $q = r^2$, where $r \geq 8$ is a power of two. 
What remains is to construct a protocol yielding a magic state for an arbitrary $q=2^n$. Following Ref.~\cite{Wills:2024wid}, we treat each $r^2$-qudit as two $r$-qudits. Then, for a $r \geq 8$, we use our codes to distill a $\ket{CCZ}_{r^2}$ state. This state can reduced with a constant-depth Clifford circuit into a $\ket{CCZ}_r$. Below, we derive this state reduction procedure and show how any $\ket{CCZ}_{2^n}$ may be obtained with this procedure.

\begin{lemma}
    \label{lemma:state-reduction}
    Given a state $\ket{CCZ}_{r^2}$, one can obtain the state $\ket{CCZ}_r$ using a Clifford circuit (Fig~\ref{fig:state-reduction}) consisting of $4$ single-qudit gates and $3$ two-qudit gates.
\end{lemma}
%
\begin{proof}
We choose a normal basis $\{\theta, \theta^r\}$ of $\mathbb{F_{r^2}}$ over $\mathbb{F}_r$, which can always be chosen such that $\theta + \theta^r = 1$. In this way, we can write
\begin{align}
    x &= x_0\theta + x_1\theta^r,\; y = y_0\theta + y_1\theta^r,\text{ and }
    z = z_0 \theta + z_1\theta^r
\end{align}
where $x_i, y_i, z_i \in \mathbb{F}_r$. 
Then, we can exploit the composition of field traces
\begin{align}
    \operatorname{Tr}_{\mathbb{F}_{r^2}/\mathbb{F}_2}(xyz) = \operatorname{Tr}_{\mathbb{F}_{r}/\mathbb{F}_2} \circ \operatorname{Tr}_{\mathbb{F}_{r^2}/\mathbb{F}_r}(xyz).
\end{align}
to rewrite 
\begin{align}
    &\ket{CCZ}_{r^2} \notag\\&\propto \sum_{x,y,z} (-1)^{\operatorname{Tr}_{\mathbb{F}_{r}/\mathbb{F}_2}\circ\operatorname{Tr}_{\mathbb{F}_{r^2}/\mathbb{F}_r}(xyz)} \ket{x_0,y_0,z_0,}_{r}\ket{x_1,y_1,z_1}_{r}.
\end{align}
and evaluate only the first trace, 
\begin{align}
    \operatorname{Tr}_{\mathbb{F}_{r^2}/\mathbb{F}_r}(xyz) &= \gamma(x_0y_0z_0 + x_1y_1z_1) + \eta \sum\limits_{i+j+k = 1,2} x_iy_jz_k\notag\\
    &\text{where }\gamma:=1+\theta^{r+1}, \: \eta:=\theta^{r+1}. 
\end{align}
Thus,
while the measurements in Fig.~\ref{fig:state-reduction} fix the values of $x_0,  y_0$, and $z_0$, the Clifford corrections cancel all phases proportional to $\eta$, so only the phase $\gamma x_1 y_1 z_1$ remains with the phase $\gamma x_0 y_0 z_0$ now a global one. 
These steps produce the state $\ket{CCZ^\gamma}_r$. 
Then, applying to register $\ket{x_1}_r$ the Clifford gate $M_\gamma$ given by
\begin{align}
    M_\gamma \ket{v}_r = \ket{\gamma v}_r \text{ for } v \in \mathbb{F}_r,
\end{align}
we recover the $\ket{CCZ}_{r}$ state.
\end{proof}

\begin{figure}
    \centering
    \includegraphics[width=\linewidth]{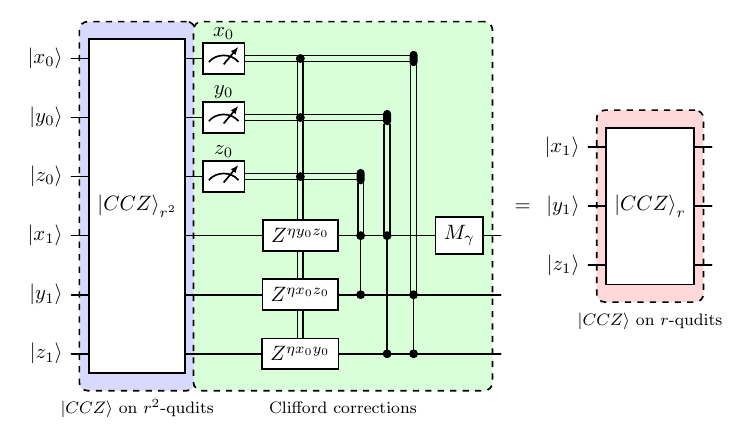}
    \caption{Reduction of state $\ket{CCZ}_{r^2}$ to a state $\ket{CCZ^{\gamma}}_r$ where $\gamma := 1 + \eta$, $\eta := \theta^{r + 1}$ and $\{\theta, \theta^r\}$ is a normal basis of $\mathbb{F}_{r^2}$ over $\mathbb{F}_r$ such that $\theta + \theta^r = 1$.}
    \label{fig:state-reduction}
\end{figure}

\begin{theorem}
    Suppose the state $\ket{CCZ}_{2^{2m}}$ for any $m \geq 3$ can be distilled (at constant space overhead with the codes of Thm.~\ref{thm:quantum-codes}). Then, any $\ket{CCZ}_{2^n}$ can be obtained with a constant-depth Clifford circuit. 
\end{theorem}
\begin{proof}

We have a few cases to consider. 
\begin{enumerate}
        \item $n$ is even and $n \geq 6$: set $r = 2^{n/2} \geq 8$ and use our codes to distill $\ket{CCZ}_{2^n}$. 
        \item $n>2$ is odd or $n=4$: set $r = 2^n \geq 8$ and distill $\ket{CCZ}_{r^2}$. Then, Fig.~\ref{fig:state-reduction} is used to obtain $\ket{CCZ}_{2^n}$.
    \item $n \leq 2$: set $r = 16$ and distill $\ket{CCZ}_{r^2}$. Then, Fig.~\ref{fig:state-reduction} is used twice to obtain $\ket{CCZ}_{4}$ or thrice for $\ket{CCZ}_{2}$. 
\end{enumerate}
Thus the reduction circuit is used at most thrice.
As a result, any $\ket{CCZ}_{2^n}$ gate may be obtained with at most 12 single-qudit and 9 two-qudit Clifford gates. 
\end{proof}
This procedure generalizes the result of Ref.~\cite{Wills:2024wid} where it was shown how a non-Clifford single-qudit gate on $2^{10}$-qudit can be converted to qubit $\ket{CCZ}$ gates. 

\section{Conclusions}
\label{sec:conclusion}

In this work, we have presented explicit constructions of low-overhead magic state distillation protocols for qudits with dimensions $2^{2m}$ with $m\geq3$ using triorthogonal codes that are asymptotically good.  In this, we have identified a new family of triorthogonal codes that satisfy the TVZ bound through a lifting method. Of the codes identified, some can be implemented on $d=2^6$ qudits -- the same dimension as in Ref.~\cite{Nguyen:2024qwg} and smaller than in Ref.~\cite{Wills:2024wid}.  These are attractive design targets for qudit-based platforms like SRF-cavities where single $d\leq100$ qudits have been demonstrated at varying levels of sophistication~\cite{Deng:2023uac,Kim:2025ywx} including a $[[42,14,6]]_{64}$ code.

Looking beyond these results, there are a number of interesting avenues of research.  First, it is critical to identify within our family of codes those best suited for current and future qudit-based platforms in more detail.  Numerics should be employed to determine the distillable states and regions with their thresholds.  Further, continued exploration of the properties of single codes or families with smaller $2^n$-dimension along the lines pursued for prime dimensions in~\cite{Saha:2025frb} would be valuable.  While they do not belong to a family of codes, we can construct a $[[25,7,\geq6]]_{32}$ and a $[[7,1,3]]_8$ codes with transversal CCZ gate, enabling protocols with low space footprint albeit higher overhead which are potentially practical on qudit-based platforms or even on qubit-based platforms with a small constant factor overhead. 
Another open question is whether qLDPC codes for qudits~\cite{Aly:2008tcy,Aly2008Families,10821254,Vasic:2025obi} can have low overhead.  Finally, recent work has suggested that mildly relaxing the requirement for asymptotically good codes can improve the locality of the weight checks~\cite{Golowich:2025utp}.

\begin{acknowledgments}
The authors thank Erik Gustafson, Di Fang, Yu Tong, and Narayanan Rengaswamy for invaluable comments in the process of this work. H.L. and E.M. were supported by the U.S. Department of Energy, Office of Science, National Quantum Information Science Research Centers, Superconducting Quantum Materials and Systems Center (SQMS), under Contract No. 89243024CSC000002. Fermilab is operated by Fermi Forward Discovery Group, LLC under Contract No. 89243024CSC000002 with the U.S. Department of Energy, Office of Science, Office of High Energy Physics. D.L. acknowledges support from the U.S. Department of Energy (DOE) under Contract No. DE-AC02-05CH11231, through the Office of Advanced Scientific Computing Research Accelerated Research for Quantum Computing Program, MACH-Q project. SZ acknowledges support from the U.S. Department of Energy, Office of Science, Accelerated Research in Quantum Computing Centers, Quantum Utility through Advanced Computational Quantum Algorithms, grant No. DE-SC0025572.
M.J.C.~is supported by the U.S.~Department of Energy Grant No.~DE-FG02-97ER-41014 (U.W.~Nuclear Theory). 
\end{acknowledgments}

\bibliography{msdbibo}

@article{HastingsHaah2018,
  title = {Distillation with Sublogarithmic Overhead},
  author = {Hastings, Matthew B. and Haah, Jeongwan},
  journal = {Phys. Rev. Lett.},
  volume = {120},
  issue = {5},
  pages = {050504},
  numpages = {3},
  year = {2018},
  month = {Jan},
  publisher = {American Physical Society},
  doi = {10.1103/PhysRevLett.120.050504},
}

@article{HowardVala2012,
  title = {Qudit versions of the qubit $\ensuremath{\pi}/8$ gate},
  author = {Howard, Mark and Vala, Jiri},
  journal = {Phys. Rev. A},
  volume = {86},
  issue = {2},
  pages = {022316},
  numpages = {10},
  year = {2012},
  month = {Aug},
  publisher = {American Physical Society},
  doi = {10.1103/PhysRevA.86.022316},
  url = {https://link.aps.org/doi/10.1103/PhysRevA.86.022316}
}

@article{nguyen2024empowering,
  title={Empowering a qudit-based quantum processor by traversing the dual bosonic ladder},
  author={Nguyen, Long B and Goss, Noah and Siva, Karthik and Kim, Yosep and Younis, Ed and Qing, Bingcheng and Hashim, Akel and Santiago, David I and Siddiqi, Irfan},
  journal={Nature Communications},
  volume={15},
  number={1},
  pages={7117},
  year={2024},
  publisher={Nature Publishing Group UK London}
}

@article{wang2020qudits,
  title={Qudits and high-dimensional quantum computing},
  author={Wang, Yuchen and Hu, Zixuan and Sanders, Barry C and Kais, Sabre},
  journal={Frontiers in Physics},
  volume={8},
  pages={589504},
  year={2020},
  url={https://www.frontiersin.org/journals/physics/articles/10.3389/fphy.2020.589504/full},
  publisher={Frontiers Media SA}
}

@article{KiktenkoNikolaevaEtAl2025,
  title = {Colloquium: Qudits for decomposing multiqubit gates and realizing quantum algorithms},
  author = {Kiktenko, Evgeniy O. and Nikolaeva, Anastasiia S. and Fedorov, Aleksey K.},
  journal = {Rev. Mod. Phys.},
  volume = {97},
  issue = {2},
  pages = {021003},
  numpages = {26},
  year = {2025},
  month = {Jun},
  publisher = {American Physical Society},
  doi = {10.1103/RevModPhys.97.021003},
  url = {https://link.aps.org/doi/10.1103/RevModPhys.97.021003}
}

@article{Chara2024,
  title = {Lifting iso-dual algebraic geometry codes},
  volume = {92},
  ISSN = {1573-7586},
  url = {http://dx.doi.org/10.1007/s10623-024-01412-y},
  DOI = {10.1007/s10623-024-01412-y},
  number = {10},
  journal = {Designs,  Codes and Cryptography},
  publisher = {Springer Science and Business Media LLC},
  author = {Chara,  María and Podestá,  Ricardo and Quoos,  Luciane and Toledano,  Ricardo},
  year = {2024},
  month = may,
  pages = {2743–2767}
}

@misc{Wang:2024xbz,
    author = "Wang, Z. and Parker, R. W. and Champion, E. and Blok, M. S.",
    title = "{Systematic study of High $E_J/E_C$ transmon qudits up to $d = 12$}",
    eprint = "2407.17407",
    archivePrefix = "arXiv",
    primaryClass = "quant-ph",
    month = "7",
    year = "2024"
}

@misc{Iiyama:2024uos,
    author = "Iiyama, Yutaro and Jang, Wonho and Kanazawa, Naoki and Sawada, Ryu and Onodera, Tamiya and Terashi, Koji",
    title = "{Qudit-Generalization of the Qubit Echo and Its Application to a Qutrit-Based Toffoli Gate}",
    eprint = "2405.14752",
    archivePrefix = "arXiv",
    primaryClass = "quant-ph",
    month = "5",
    year = "2024"
}

@article{Omanakuttan:2021ffn,
    author = "Omanakuttan, Sivaprasad and Mitra, Anupam and Martin, Michael J. and Deutsch, Ivan H.",
    title = "{Quantum optimal control of ten-level nuclear spin qudits in Sr87}",
    eprint = "2106.13705",
    archivePrefix = "arXiv",
    primaryClass = "quant-ph",
    reportNumber = "LA-UR-21-25988",
    doi = "10.1103/PhysRevA.104.L060401",
    journal = "Phys. Rev. A",
    volume = "104",
    number = "6",
    pages = "L060401",
    year = "2021"
}

@misc{Vasic:2025obi,
    author = "Vasic, Bane and Savin, Valentin and Pacenti, Michele and Borah, Shantom and Raveendran, Nithin",
    title = "{Quantum Low-Density Parity-Check Codes}",
    eprint = "2510.14090",
    archivePrefix = "arXiv",
    primaryClass = "quant-ph",
    month = "10",
    year = "2025"
}

@INPROCEEDINGS{10821254,
  author={Borah, Shantom K. and Pradhan, Asit K. and Raveendran, Nithin and Rengaswamy, Narayanan and Vasić, Bane},
  booktitle={2024 IEEE International Conference on Quantum Computing and Engineering (QCE)}, 
  title={Non-Binary Hypergraph Product Codes for Qudit Error Correction}, 
  year={2024},
  volume={01},
  number={},
  pages={98-108},
  doi={10.1109/QCE60285.2024.00021}}

@article{deFuentes:2023pbp,
    author = "de Fuentes, Irene Fern{\'a}ndez and others",
    title = "{Navigating the 16-dimensional Hilbert space of a high-spin donor qudit with electric and magnetic fields}",
    eprint = "2306.07453",
    archivePrefix = "arXiv",
    primaryClass = "quant-ph",
    doi = "10.1038/s41467-024-45368-y",
    journal = "Nature Commun.",
    volume = "15",
    number = "1",
    pages = "1380",
    year = "2024"
}

@misc{Shi:2025vvq,
    author = "Shi, Xiaoyang and Sinanan-Singh, Jasmine and Burke, Timothy J. and Chiaverini, John and Chuang, Isaac L.",
    title = "{Efficient Implementation of a Quantum Algorithm with a Trapped Ion Qudit}",
    eprint = "2506.09371",
    archivePrefix = "arXiv",
    primaryClass = "quant-ph",
    month = "6",
    year = "2025"
}

@article{Champion:2024wlh,
    author = "Champion, Elizabeth and Wang, Zihao and Parker, Rayleigh W. and Blok, Machiel S.",
    title = "{Efficient Control of a Transmon Qudit Using Effective Spin-7/2 Rotations}",
    eprint = "2405.15857",
    archivePrefix = "arXiv",
    primaryClass = "quant-ph",
    doi = "10.1103/vbh4-lysv",
    journal = "Phys. Rev. X",
    volume = "15",
    number = "2",
    pages = "021096",
    year = "2025"
}

@article{Dong:2023xhv,
    author = "Dong, Ming-Xin and Zhang, Wei-Hang and Zeng, Lei and Ye, Ying-Hao and Li, Da-Chuang and Guo, Guang-Can and Ding, Dong-Sheng and Shi, Bao-Sen",
    title = "{Highly Efficient Storage of 25-Dimensional Photonic Qudit in a Cold-Atom-Based Quantum Memory}",
    eprint = "2301.00999",
    archivePrefix = "arXiv",
    primaryClass = "quant-ph",
    doi = "10.1103/PhysRevLett.131.240801",
    journal = "Phys. Rev. Lett.",
    volume = "131",
    number = "24",
    pages = "240801",
    year = "2023"
}

@misc{Mezzadri:2023ige,
    author = "Mezzadri, Matteo and Chiesa, Alessandro and Lepori, Luca and Carretta, Stefano",
    title = "{Fault-Tolerant Computing with Single Qudit Encoding}",
    eprint = "2307.10761",
    archivePrefix = "arXiv",
    primaryClass = "quant-ph",
    doi = "10.1039/D4MH00454J",
    month = "7",
    year = "2023"
}

@article{Deng:2023uac,
    author = "Deng, Xiaowei and others",
    title = "{Quantum-enhanced metrology with large Fock states}",
    eprint = "2306.16919",
    archivePrefix = "arXiv",
    primaryClass = "quant-ph",
    doi = "10.1038/s41567-024-02619-5",
    journal = "Nature Phys.",
    volume = "20",
    number = "12",
    pages = "1874--1880",
    year = "2024"
}

@article{Nikolaeva:2021rhq,
    author = "Nikolaeva, Anastasiia S. and Kiktenko, Evgeniy O. and Fedorov, Aleksey K.",
    title = "{Efficient realization of quantum algorithms with qudits}",
    eprint = "2111.04384",
    archivePrefix = "arXiv",
    primaryClass = "quant-ph",
    doi = "10.1140/epjqt/s40507-024-00250-0",
    journal = "EPJ Quant. Technol.",
    volume = "11",
    number = "1",
    pages = "43",
    year = "2024"
}

@misc{Low:2023dlg,
    author = "Low, Pei Jiang and White, Brendan and Senko, Crystal",
    title = "{Control and Readout of a 13-level Trapped Ion Qudit}",
    eprint = "2306.03340",
    archivePrefix = "arXiv",
    primaryClass = "quant-ph",
    month = "6",
    year = "2023"
}

@article{Ringbauer:2021lhi,
    author = "Ringbauer, Martin and Meth, Michael and Postler, Lukas and Stricker, Roman and Blatt, Rainer and Schindler, Philipp and Monz, Thomas",
    title = "{A universal qudit quantum processor with trapped ions}",
    eprint = "2109.06903",
    archivePrefix = "arXiv",
    primaryClass = "quant-ph",
    doi = "10.1038/s41567-022-01658-0",
    journal = "Nature Phys.",
    volume = "18",
    number = "9",
    pages = "1053--1057",
    year = "2022"
}

@misc{chara2025goodisodualagcodestowers,
      title={Good iso-dual AG-codes from towers of function fields}, 
      author={María Chara and Ricardo Podestá and Luciane Quoos and Ricardo Toledano},
      year={2025},
      eprint={2503.08899},
      archivePrefix={arXiv},
      primaryClass={math.NT}
}

@book{Stichtenoth2009,
  title = {Algebraic Function Fields and Codes},
  ISBN = {9783540768784},
  ISSN = {0072-5285},
  url = {http://dx.doi.org/10.1007/978-3-540-76878-4},
  DOI = {10.1007/978-3-540-76878-4},
  journal = {Graduate Texts in Mathematics},
  publisher = {Springer Berlin Heidelberg},
  author = {Stichtenoth,  Henning},
  year = {2009}
}

@book{agcodes-advanced,
  title     = "Applications",
  author={Tsfasman, Michael and Vlǎdu{\c{t}}, Serge and Nogin, Dmitry},
  booktitle = "Algebraic Geometry Codes: Advanced Chapters",
  publisher = "American Mathematical Society",
  pages     = "335--396",
  series    = "Mathematical Surveys and Monographs",
  month     =  jul,
  year      =  2019,
  address   = "Providence, Rhode Island"
}

@article{Wills:2024wid,
    author = "Wills, Adam and Hsieh, Min-Hsiu and Yamasaki, Hayata",
    title = "{Constant-overhead magic state distillation}",
    journal = "Nature Phys.",
    year = "2025",
    doi = {10.1038/s41567-025-03026-0},
    issn = "1745-2481"
}

@article{STICHTENOTH1988199,
title = {Self-dual Goppa codes},
journal = {Journal of Pure and Applied Algebra},
volume = {55},
number = {1},
pages = {199-211},
year = {1988},
issn = {0022-4049},
doi = {https://doi.org/10.1016/0022-4049(88)90046-1},
url = {https://www.sciencedirect.com/science/article/pii/0022404988900461},
author = {Henning Stichtenoth}
}

@article{PhysRevLett.123.070507,
  title = {Towards Low Overhead Magic State Distillation},
  author = {Krishna, Anirudh and Tillich, Jean-Pierre},
  journal = {Phys. Rev. Lett.},
  volume = {123},
  issue = {7},
  pages = {070507},
  numpages = {4},
  year = {2019},
  month = {Aug},
  publisher = {American Physical Society},
  doi = {10.1103/PhysRevLett.123.070507},
  url = {https://link.aps.org/doi/10.1103/PhysRevLett.123.070507}
}

@article{Varshamov1957,
  author    = {Varshamov, R. R.},
  title     = {Estimate of the Number of Signals in Error Correcting Codes},
  journal   = {Doklady Akademii Nauk SSSR},
  volume    = {117},
  pages     = {739--741},
  year      = {1957},
}

@ARTICLE{6773017,
  author={Gilbert, E. N.},
  journal={The Bell System Technical Journal}, 
  title={A comparison of signalling alphabets}, 
  year={1952},
  volume={31},
  number={3},
  pages={504-522},
  doi={10.1002/j.1538-7305.1952.tb01393.x}}

@article{mana.19821090103,
author = {Tsfasman, M. A. and Vlădutx, S. G. and Zink, Th.},
title = "{{M}odular curves, {S}himura curves, and {G}oppa codes, better than {V}arshamov-{G}ilbert bound}",
journal = {Mathematische Nachrichten},
volume = {109},
number = {1},
pages = {21-28},
doi = {10.1002/mana.19821090103},
year = {1982}
}

@misc{Aly2008Families,
  author       = {Aly, Salah A.},
  title        = {Families of {LDPC} Codes Derived from Nonprimitive {BCH} Codes and Cyclotomic Cosets},
  eprint       = {0802.4079},
  archivePrefix= {arXiv},
  primaryClass = {cs.IT},
  year         = {2008}
}

@mastersthesis{Aly:2008tcy,
    author = "Aly, Salah A.",
    title = "{On Quantum and Classical Error Control Codes: Constructions and Applications}",
    eprint = "0812.5104",
    archivePrefix = "arXiv",
    primaryClass = "cs.IT",
    school ="Texas A\&M University",
    type = "Dissertation",
    month = "12",
    year = "2008"
}

@misc{Kim:2025ywx,
    author = "Kim, Taeyoon and others",
    title = "{Ultracoherent superconducting cavity-based multiqudit platform with error-resilient control}",
    eprint = "2506.03286",
    archivePrefix = "arXiv",
    primaryClass = "quant-ph",
    reportNumber = "FERMILAB-PUB-25-0379-SQMS",
    month = "6",
    year = "2025"
}

@article{Jankovic:2023bnw,
    author = "Jankov{\'\i}c, Denis and Hartmann, Jean-Gabriel and Ruben, Mario and Hervieux, Paul-Antoine",
    title = "{Noisy qudit vs multiple qubits: conditions on gate efficiency for enhancing fidelity}",
    eprint = "2302.04543",
    archivePrefix = "arXiv",
    primaryClass = "quant-ph",
    doi = "10.1038/s41534-024-00829-6",
    journal = "npj Quantum Inf.",
    volume = "10",
    number = "1",
    pages = "59",
    year = "2024"
}

@misc{Keppens:2025pfo,
    author = "Keppens, James and Eggerickx, Quinten and Levajac, Vukan and Simion, George and Sor{\'e}e, Bart",
    title = "{Qudit vs. Qubit: Simulated performance of error correction codes in higher dimensions}",
    eprint = "2502.05992",
    archivePrefix = "arXiv",
    primaryClass = "quant-ph",
    month = "2",
    year = "2025"
}

@article{Campbell:2012olh,
    author = "Campbell, Earl T. and Anwar, Hussain and Browne, Dan E.",
    title = "{Magic-State Distillation in All Prime Dimensions Using Quantum Reed-Muller Codes}",
    doi = "10.1103/PhysRevX.2.041021",
    journal = "Phys. Rev. X",
    volume = "2",
    number = "4",
    pages = "041021",
    year = "2012"
}

@article{PhysRevA.71.022316,
  title = {Universal quantum computation with ideal Clifford gates and noisy ancillas},
  author = {Bravyi, Sergey and Kitaev, Alexei},
  journal = {Phys. Rev. A},
  volume = {71},
  issue = {2},
  pages = {022316},
  numpages = {14},
  year = {2005},
  month = {Feb},
  publisher = {American Physical Society},
  doi = {10.1103/PhysRevA.71.022316}
}

@misc{perlin2023qldpc,
  author = {Perlin, Michael A.},
  title = {{qLDPC}},
  year = {2023},
  publisher = {GitHub},
  journal = {GitHub repository},
  howpublished = {\url{https://github.com/qLDPCOrg/qLDPC}},
}

@misc{Golowich:2025utp,
    author = "Golowich, Louis and Guruswami, Venkatesan",
    title = "{Near-Asymptotically-Good Quantum Codes with Transversal CCZ Gates and Sublinear-Weight Parity-Checks}",
    eprint = "2510.06798",
    archivePrefix = "arXiv",
    primaryClass = "quant-ph",
    month = "10",
    year = "2025"
}

@article{Mansky:2022bai,
    author = "Mansky, Maximilian Balthasar and Castillo, Santiago Londo\~no and Puigvert, Victor Ramos and Linnhoff-Popien, Claudia",
    title = "{Near-optimal quantum circuit construction via Cartan decomposition}",
    eprint = "2212.12934",
    archivePrefix = "arXiv",
    primaryClass = "quant-ph",
    doi = "10.1103/PhysRevA.108.052607",
    journal = "Phys. Rev. A",
    volume = "108",
    number = "5",
    pages = "052607",
    year = "2023"
}

@article{kiktenko2020scalable,
  title = {Scalable quantum computing with qudits on a graph},
  author = {Kiktenko, E. O. and Nikolaeva, A. S. and Xu, Peng and Shlyapnikov, G. V. and Fedorov, A. K.},
  journal = {Phys. Rev. A},
  volume = {101},
  issue = {2},
  pages = {022304},
  numpages = {7},
  year = {2020},
  month = {Feb},
  publisher = {American Physical Society},
  doi = {10.1103/PhysRevA.101.022304},
  url = {https://link.aps.org/doi/10.1103/PhysRevA.101.022304}
}

@article{Nikolaeva:2022wmq,
    author = "Nikolaeva, Anastasiia S. and Kiktenko, Evgeniy O. and Fedorov, Aleksey K.",
    title = "{Generalized Toffoli Gate Decomposition Using Ququints: Towards Realizing Grover\textquoteright{}s Algorithm with Qudits}",
    eprint = "2212.12505",
    archivePrefix = "arXiv",
    primaryClass = "quant-ph",
    doi = "10.3390/e25020387",
    journal = "Entropy",
    volume = "25",
    number = "2",
    pages = "387",
    year = "2023"
}

@misc{Champion:2024ufp,
    author = "Champion, Elizabeth and Wang, Zihao and Parker, Rayleigh and Blok, Machiel",
    title = "{Multi-frequency control and measurement of a spin-7/2 system encoded in a transmon qudit}",
    eprint = "2405.15857",
    archivePrefix = "arXiv",
    primaryClass = "quant-ph",
    month = "5",
    year = "2024"
}

@article{BassmanOftelie:2022hfz,
    author = "Bassman Oftelie, Lindsay and Klymko, Katherine and Liu, Diyi and Tubman, Norm M. and de Jong, Wibe A.",
    title = "{Computing Free Energies with Fluctuation Relations on Quantum Computers}",
    doi = "10.1103/PhysRevLett.129.130603",
    journal = "Phys. Rev. Lett.",
    volume = "129",
    number = "13",
    pages = "130603",
    year = "2022"
}

@misc{Joshi:2025pgv,
    author = "Joshi, Rohan and Meth, Michael and Louw, Jan C. and Osborne, Jesse J. and Mato, Kevin and Ringbauer, Martin and Halimeh, Jad C.",
    title = "{Efficient Qudit Circuit for Quench Dynamics of $2+1$D Quantum Link Electrodynamics}",
    eprint = "2507.12589",
    archivePrefix = "arXiv",
    primaryClass = "quant-ph",
    month = "7",
    year = "2025"
}

@misc{Kurkcuoglu:2024cfv,
    author = {K{\"u}rk{\c{c}}{\"u}oglu, Doga Murat and Lamm, Henry and Maestri, Andrea},
    title = "{Qudit Gate Decomposition Dependence for Lattice Gauge Theories}",
    eprint = "2410.16414",
    archivePrefix = "arXiv",
    primaryClass = "quant-ph",
    reportNumber = "FERMILAB-PUB-24-0612-SQMS-T",
    month = "10",
    year = "2024"
}

@article{Meth:2023wzd,
    author = "Meth, Michael and others",
    title = "{Simulating two-dimensional lattice gauge theories on a qudit quantum computer}",
    eprint = "2310.12110",
    archivePrefix = "arXiv",
    primaryClass = "quant-ph",
    doi = "10.1038/s41567-025-02797-w",
    journal = "Nature Phys.",
    volume = "21",
    number = "4",
    pages = "570--576",
    year = "2025"
}

@article{Zache:2023cfj,
    author = "Zache, Torsten V. and Gonz\'alez-Cuadra, Daniel and Zoller, Peter",
    title = "{Fermion-qudit quantum processors for simulating lattice gauge theories with matter}",
    eprint = "2303.08683",
    archivePrefix = "arXiv",
    primaryClass = "quant-ph",
    doi = "10.22331/q-2023-10-16-1140",
    journal = "Quantum",
    volume = "7",
    pages = "1140",
    year = "2023"
}

@article{Illa:2024kmf,
    author = "Illa, Marc and Robin, Caroline E. P. and Savage, Martin J.",
    title = "{Qu8its for quantum simulations of lattice quantum chromodynamics}",
    eprint = "2403.14537",
    archivePrefix = "arXiv",
    primaryClass = "quant-ph",
    reportNumber = "IQuS@UW-21-074",
    doi = "10.1103/PhysRevD.110.014507",
    journal = "Phys. Rev. D",
    volume = "110",
    number = "1",
    pages = "014507",
    year = "2024"
}

@misc{Calajo:2024qrc,
    author = "Calaj\`o, Giuseppe and Magnifico, Giuseppe and Edmunds, Claire and Ringbauer, Martin and Montangero, Simone and Silvi, Pietro",
    title = "{Digital quantum simulation of a (1+1)D SU(2) lattice gauge theory with ion qudits}",
    eprint = "2402.07987",
    archivePrefix = "arXiv",
    primaryClass = "quant-ph",
    month = "2",
    year = "2024"
}

@article{Popov:2023xft,
    author = "Popov, Pavel P. and Meth, Michael and Lewenstein, Maciej and Hauke, Philipp and Ringbauer, Martin and Zohar, Erez and Kasper, Valentin",
    title = "{Variational quantum simulation of U(1) lattice gauge theories with qudit systems}",
    doi = "10.1103/PhysRevResearch.6.013202",
    journal = "Phys. Rev. Res.",
    volume = "6",
    number = "1",
    pages = "013202",
    year = "2024"
}

@misc{Gonzalez-Cuadra:2022hxt,
    author = "Gonz\'alez-Cuadra, Daniel and Zache, Torsten V. and Carrasco, Jose and Kraus, Barbara and Zoller, Peter",
    title = "{Hardware efficient quantum simulation of non-abelian gauge theories with qudits on Rydberg platforms}",
    eprint = "2203.15541",
    archivePrefix = "arXiv",
    primaryClass = "quant-ph",
    month = "3",
    year = "2022"
}

@article{Gustafson:2021qbt,
    author = "Gustafson, Erik",
    title = "{Prospects for Simulating a Qudit Based Model of (1+1)d Scalar QED}",
    eprint = "2104.10136",
    archivePrefix = "arXiv",
    primaryClass = "quant-ph",
    doi = "10.1103/PhysRevD.103.114505",
    journal = "Phys. Rev. D",
    volume = "103",
    number = "11",
    pages = "114505",
    year = "2021"
}

@inproceedings{Young:2023zxu,
    author = "Young, Amanda",
    title = "{Quantum Spin Systems}",
    eprint = "2308.07848",
    archivePrefix = "arXiv",
    primaryClass = "math-ph",
    month = "8",
    year = "2023"
}

@article{choi2017dynamical,
  title = {Dynamical Engineering of Interactions in Qudit Ensembles},
  author = {Choi, Soonwon and Yao, Norman Y. and Lukin, Mikhail D.},
  journal = {Phys. Rev. Lett.},
  volume = {119},
  issue = {18},
  pages = {183603},
  numpages = {6},
  year = {2017},
  month = {Nov},
  publisher = {American Physical Society},
  doi = {10.1103/PhysRevLett.119.183603},
  url = {https://link.aps.org/doi/10.1103/PhysRevLett.119.183603}
}

@article{PhysRevA.110.062602,
  title = {Qudit-based quantum simulation of fermionic systems},
  author = {Chizzini, M. and Tacchino, F. and Chiesa, A. and Tavernelli, I. and Carretta, S. and Santini, P.},
  journal = {Phys. Rev. A},
  volume = {110},
  issue = {6},
  pages = {062602},
  numpages = {14},
  year = {2024},
  month = {Dec},
  publisher = {American Physical Society},
  doi = {10.1103/PhysRevA.110.062602},
  url = {https://link.aps.org/doi/10.1103/PhysRevA.110.062602}
}

@article{Bravyi:2020pur,
    author = "Bravyi, Sergey and Kliesch, Alexander and Koenig, Robert and Tang, Eugene",
    title = "{Hybrid quantum-classical algorithms for approximate graph coloring}",
    eprint = "2011.13420",
    archivePrefix = "arXiv",
    primaryClass = "quant-ph",
    doi = "10.22331/q-2022-03-30-678",
    journal = "Quantum",
    volume = "6",
    pages = "678",
    year = "2022"
}

@article{BravyiKitaev05,
  title = {Universal quantum computation with ideal Clifford gates and noisy ancillas},
  author = {Bravyi, Sergey and Kitaev, Alexei},
  journal = {Phys. Rev. A},
  volume = {71},
  issue = {2},
  pages = {022316},
  numpages = {14},
  year = {2005},
  month = {Feb},
  publisher = {American Physical Society},
  doi = {10.1103/PhysRevA.71.022316},
  url = {https://link.aps.org/doi/10.1103/PhysRevA.71.022316}
}

@misc{Tancara:2024vck,
    author = "Tancara, Diego and Albarr\'an-Arriagada, Francisco",
    title = "{High-dimensional counterdiabatic quantum computing}",
    eprint = "2410.10622",
    archivePrefix = "arXiv",
    primaryClass = "quant-ph",
    month = "10",
    year = "2024"
}

@article{zhu2024unified,
  title = {Unified architecture for quantum lookup tables},
  author = {Zhu, Shuchen and Sundaram, Aarthi and Low, Guang Hao},
  journal = {Phys. Rev. Res.},
  volume = {7},
  issue = {4},
  pages = {043230},
  numpages = {20},
  year = {2025},
  month = {Dec},
  publisher = {American Physical Society},
  doi = {10.1103/d896-mktn},
  url = {https://link.aps.org/doi/10.1103/d896-mktn}
}

@article{PRXQuantum.2.020311,
  title = {Resilience of Quantum Random Access Memory to Generic Noise},
  author = {Hann, Connor T. and Lee, Gideon and Girvin, S.M. and Jiang, Liang},
  journal = {PRX Quantum},
  volume = {2},
  issue = {2},
  pages = {020311},
  numpages = {30},
  year = {2021},
  month = {Apr},
  publisher = {American Physical Society},
  doi = {10.1103/PRXQuantum.2.020311},
  url = {https://link.aps.org/doi/10.1103/PRXQuantum.2.020311}
}

@article{roy2023two,
  title = {Two-Qutrit Quantum Algorithms on a Programmable Superconducting Processor},
  author = {Roy, Tanay and Li, Ziqian and Kapit, Eliot and Schuster, DavidI.},
  journal = {Phys. Rev. Appl.},
  volume = {19},
  issue = {6},
  pages = {064024},
  numpages = {18},
  year = {2023},
  month = {Jun},
  publisher = {American Physical Society},
  doi = {10.1103/PhysRevApplied.19.064024},
  url = {https://link.aps.org/doi/10.1103/PhysRevApplied.19.064024}
}

@misc{Murairi:2024xpc,
    author = "Murairi, Edison M. and Sohaib Alam, M. and Lamm, Henry and Hadfield, Stuart and Gustafson, Erik",
    title = "{Highly-efficient quantum Fourier transformations for some nonabelian groups}",
    eprint = "2408.00075",
    archivePrefix = "arXiv",
    primaryClass = "quant-ph",
    reportNumber = "FERMILAB-PUB-24-0241-SQMS-T",
    month = "7",
    year = "2024"
}

@misc{Saha:2025frb,
    author = "Saha, Tanay and Prakash, Shiroman",
    title = "{Sublogarithmic Distillation in all Prime Dimensions using Punctured Reed-Muller Codes}",
    eprint = "2510.10852",
    archivePrefix = "arXiv",
    primaryClass = "quant-ph",
    month = "10",
    year = "2025"
}

@inproceedings{GolowichGuruswami2024,
    author = "Golowich, Louis and Guruswami, Venkatesan",
    title = "{Asymptotically Good Quantum Codes with Transversal Non-Clifford Gates}",
    booktitle = "{57th Annual ACM Symposium on Theory of Computing}",
    eprint = "2408.09254",
    archivePrefix = "arXiv",
    primaryClass = "quant-ph",
    doi = "10.1145/3717823.3718234",
    month = "8",
    year = "2024"
}

@inproceedings{Nguyen:2024qwg,
author = {Nguyen, Quynh T.},
title = {Good Binary Quantum Codes with Transversal CCZ Gate},
year = {2025},
isbn = {9798400715105},
publisher = {Association for Computing Machinery},
address = {New York, NY, USA},
url = {https://doi.org/10.1145/3717823.3718186},
doi = {10.1145/3717823.3718186},
booktitle = {Proceedings of the 57th Annual ACM Symposium on Theory of Computing},
pages = {697–706},
numpages = {10},
location = {Prague, Czechia},
series = {STOC '25}
}

@article{PhysRevA.86.052329,
  title = {Magic-state distillation with low overhead},
  author = {Bravyi, Sergey and Haah, Jeongwan},
  journal = {Phys. Rev. A},
  volume = {86},
  issue = {5},
  pages = {052329},
  numpages = {10},
  year = {2012},
  month = {Nov},
  publisher = {American Physical Society},
  doi = {10.1103/PhysRevA.86.052329},
  url = {https://link.aps.org/doi/10.1103/PhysRevA.86.052329}
}

@misc{Prakash:2024ghq,
    author = "Prakash, Shiroman and Singhal, Rishabh",
    title = "{A Search for High-Threshold Qutrit Magic State Distillation Routines}",
    eprint = "2408.00436",
    archivePrefix = "arXiv",
    primaryClass = "quant-ph",
    month = "8",
    year = "2024"
}

@article{Sharma:2024zbn,
    author = "Sharma, Abhi Kumar and Garani, Shayan Srinivasa",
    title = "{Near-threshold qudit stabilizer codes with efficient encoding circuits for magic-state distillation}",
    doi = "10.1103/PhysRevA.109.062426",
    journal = "Phys. Rev. A",
    volume = "109",
    number = "6",
    pages = "062426",
    year = "2024",
    note = "[Erratum: Phys.Rev.A 110, 049902 (2024)]"
}

@article{Prakash:2025azi,
    author = "Prakash, Shiroman and Saha, Tanay",
    title = "{Low Overhead Qutrit Magic State Distillation}",
    eprint = "2403.06228",
    archivePrefix = "arXiv",
    primaryClass = "quant-ph",
    doi = "10.22331/q-2025-06-12-1768",
    journal = "Quantum",
    volume = "9",
    pages = "1768",
    year = "2025"
}

@misc{Knill:2004ctr,
    author = "Knill, E.",
    title = "{Fault-Tolerant Postselected Quantum Computation: Schemes}",
    eprint = "quant-ph/0402171",
    archivePrefix = "arXiv",
    doi = "10.48550/arXiv.quant-ph/0402171",
    month = "2",
    year = "2004"
}

@article{m-gretchen, author = {Matthews, Gretchen L. and Morrison, Travis and Murphy, Aidan W.}, title = {Curve-lifted codes for local recovery using lines}, year = {2024}, issue_date = {Nov 2024}, publisher = {Kluwer Academic Publishers}, address = {USA}, volume = {92}, number = {11}, issn = {0925-1022}, url = {https://doi.org/10.1007/s10623-024-01456-0}, doi = {10.1007/s10623-024-01456-0}, journal = {Des. Codes Cryptography}, month = jul, pages = {3645–3664}, numpages = {20}
}

@article{DOUGHERTY2005123,
title = {Lifted codes and their weight enumerators},
journal = {Discrete Mathematics},
volume = {305},
number = {1},
pages = {123-135},
year = {2005},
issn = {0012-365X},
doi = {https://doi.org/10.1016/j.disc.2005.08.004},
url = {https://www.sciencedirect.com/science/article/pii/S0012365X05005066},
author = {Steven T. Dougherty and Sun Young Kim and Young Ho Park}
}

@ARTICLE{10931129,
  author={Guemard, Virgile},
  journal={IEEE Transactions on Information Theory}, 
  title={Lifts of Quantum CSS Codes}, 
  year={2025},
  volume={71},
  number={7},
  pages={5418-5442},
  doi={10.1109/TIT.2025.3552211}}

@article{Old_2024,
doi = {10.1088/2058-9565/ad5eb6},
url = {https://doi.org/10.1088/2058-9565/ad5eb6},
year = {2024},
month = {jul},
publisher = {IOP Publishing},
volume = {9},
number = {4},
pages = {045012},
author = {Old, Josias and Rispler, Manuel and Müller, Markus},
title = {Lift-connected surface codes},
journal = {Quantum Science and Technology}
}

\appendix

\section{\texorpdfstring{Canonical divisor $(\eta_0)$}{Canonical divisor eta0}}
\label{app:divisor}

In this section, we compute the canonical divisor $(\eta_0)$ for 
\begin{align}
    \eta_0 := \frac{dz_0}{z_0} && \text{ with } z_0 := \prod\limits_{\alpha \in \mathcal{Z}} (x_0 + \alpha),
\end{align}
where $Z$ are the rational places in $F_0=\mathbb{F}_q(x_0)$. 
To perform this computation, we observe that for a finite field $\mathbb{F}_q$, 
\begin{align}
    \prod_{\alpha\in \mathbb{F}_q} (x - \alpha) = x^q - x
\end{align}
and that $\mathbb{F}_{r^2} \backslash \mathcal{Z} = \mathbb{F}_r$. Then, noting that our fields have characteristics 2, we find that 
\begin{align}
    z_0 &= \frac{x^{r^2}_0 + x_0}{x^r_0 + x_0}.
\end{align}
where we have used the fact that $\mathbb{F}_{r^2}\backslash \mathcal{Z} = \mathbb{F}_{r}$. Computing the differential $dz_0$,
\begin{align}
    dz_0 &= \frac{x^{r^2}_0 + x^r_0}{(x^r_0+x_0)^2} dx_0\notag\\
    &= \frac{(x^r_0 + x_0)^r}{(x^r_0+x_0)^2} dx_0\notag\\
    &= (x^r_0 + x_0)^{r-2} dx_0\notag\\
    &= \prod\limits_{\alpha \in \mathbb{F}_r} (x_0 + \alpha)^{r-2} \; dx_0.
\end{align}
So far, we have used only ordinary calculus and recognized terms multiplied by $r$ to be trivial. Putting the canonical divisor together, 
\begin{align}
    \eta_0 &= \frac{dz_0}{z_0} =  \frac{\prod\limits_{\alpha \in \mathbb{F}_r} (x_0 + \alpha)^{r-2}}{\prod\limits_{\alpha \in \mathcal{Z}} (x_0 + \alpha)} dx_0,
\end{align}
it is clear now that
\begin{itemize}
    \item The places $P_\alpha$ for $\alpha \in \mathcal{Z}$ are simple poles of $\eta_0$ with multiplicities 1. 
    \item Each place $Q_\alpha$ for $\alpha \in \mathbb{F_r}$ is a zero of $\eta_0$, each with multiplicity $r - 2$. 
    \item The place at infinity $Q_\infty$ has valuation 
    \begin{align}
        v_{Q_\infty}(\eta_0) &= (r^2- r) - r(r-2) - 2\notag\\
        &= r - 2.
    \end{align}
\end{itemize}
Thus, the divisor $(\eta_0)$ is
\begin{align}
    (\eta_0) &= -D_0 + (r-2) \sum\limits_{\alpha \in \mathbb{F_r}} Q_\alpha + (r-2)Q_\infty\notag\\
    &= -D_0 + (r-2) \sum\limits_{Q \in \mathcal{V}} Q.
\end{align}

\section{Alternative Lifting Procedure}
\label{app:alt-lift}

In Sec.~\ref{sec:classical-code}, we demonstrate how a base AG code may be lifted repeatedly to generate a family of AG codes explicitly with good parameters. 
Here, we outline a similar lifting method, which offers some improvements upon these parameters, although their asymptotic behavior is the same. 

This alternative lifting method is similar to that of Ref.~\cite{Chara2024}. 
As in our earlier method, the definition and the lifting of our first divisor $D_j$ remain the same. 
Instead, we may lift the $G_j$ divisor by 
\begin{align}
    G_{j+1} := \operatorname{Con}_{F_{j+1}/F_j} (G_j) + \Delta(F_{j+1}/F_j)
\end{align}
where
\begin{align}
    \Delta(F_{j+1}/F_j) :=  \sum\limits_{P \in \mathbb{P}(F_j)} \sum\limits_{P'|P} \left\lfloor \frac{d(P'|P)}{3} \right\rfloor P'
\end{align}
and $d(P'|P)$ is the different exponent of $P'|P$, defined as follows:
\begin{definition}[Different exponent]
    The different exponent of a place 
        $P'|P$ 
        is $d(P'|P)=-v_{P'}(t)\geq0$ for $t$ forming the complementary module $t O_P'=C_P:=\{x\in F'\mid \mathrm{Tr}_{F'/F}(xO_P^\prime)\subseteq\mathcal{O}_P\}$ of the integral closure $O_P':=\bigcap_{P'|P}\mathcal{O}_{P'}$ of $O_P$.
\end{definition}
The divisors $D_{j+1}$ and $G_{j+1}$ are disjoint by construction, because a place $P' \in \operatorname{supp}D_{j+1}$ has $d(P'|P) = e(P'|P) - 1 = 0$ by Dedekind's different formula (Thm. 3.5.1 in Ref.~\cite{Stichtenoth2009}). 
Therefore, this lifting provides a valid AG code $C_{j+1} := C_\mathcal{L}(D_{j+1}, G_{j+1})$. 
Crucially, it is also triorthogonality-preserving since
\begin{align}
    3 \Delta(F_{j+1}/F_j) \leq \operatorname{Diff}(F_{j+1}/F_j),
\end{align}
where the different divisor Diff is related to the different exponent via the definition: 
    \begin{equation}
        \mathrm{Diff}(F'/F):=\sum_{P}\sum_{P'|P} d(P'|P)  P'.
    \end{equation}

\section{Proof of Good Quantum Codes' Parameters (Thm.~\ref{thm:quantum-codes}) }
\label{app:quantum-codes}
To prove this theorem, we will explicitly derived the quantum codes from our triorthogonal codes following the discussion in Sec.~\ref{sec:quantum-codes}. First, we note that Eq.~\eqref{eq:quantum-dist-bound1},
\begin{align}
    \label{eq:quantum-dist-bound2}
    \mathcal{D}_j \geq \operatorname{deg}G_j - \mathcal{K}_j - (2g_j - 2)
\end{align}
can be verified simply by comparing the lower bounds on $\mathcal{D}_X$ and $\mathcal{D}_Z$. Then, in light of Eq.~\eqref{eq:quantum-dist-bound2}, we must ensure
\begin{align}
    \label{eq:Kj-cond}
    \mathcal{K}_j &\leq \operatorname{deg}(G_j) + 2(1 - g_j),
\end{align}
which we may achieve simply by requiring
\begin{align}
    \mathcal{K}_j < r^j\left( (r+1) \left\lfloor \frac{r-2}{3} \right\rfloor -2r\right) + v,
\end{align}
where $v(r,j)$ is a positive function given by
\begin{align}
    v(r,j) := \begin{cases}
        4 r^{\frac{j+1}{2}}\, \text{ if $j$ is odd}\notag\\
        2r^{\frac{j}{2}} \left(1 + r\right)\, \text{ otherwise}
    \end{cases}.
\end{align}
(Recall here that $\mathrm{deg}G_j = r^j\mathrm{deg}G_0 = r^j (r+1) \left\lfloor \frac{r-2}{3} \right\rfloor$, and our genus $g_j$ is given in terms of our dimension $r$ by Eq.~\eqref{eq:1-genus}.) 
Fixing $r \geq 8$, our strategy is to equivalently choose 
\begin{align}
    \mathcal{K}_j = x_1 \, r^j \left( (r+1) \left\lfloor \frac{r-2}{3} \right\rfloor -2r\right) + x_2 \, v(r,j),
\end{align}
where $x_1,x_2$ are constants such that $0 < x_1 < 1 $, $0\leq x_2 \leq  1$ and that $\mathcal{K}_j < k_j$ is an integer. Clearly, such constants always exist for $r \geq 8$, e.g. $x_1 = 1/2$ and $x_2 = 0$.

Having chosen $\mathcal{K}_j$, it follows then that
\begin{align}
    \mathcal{N}_j &= n_j - \mathcal{K}_j\notag\\
    &= r^j \left[r^2 - r - x_1  \left( (r+1) \left\lfloor \frac{r-2}{3} \right\rfloor -2r\right) \right]\notag\\ &\;\;\;\;\; - x_2\,v(r,j)
\end{align}
and the distance satsifies
\begin{align}
    \mathcal{D}_j &\geq (1-x_1)\, r^j \, \left( (r+1) \left\lfloor \frac{r-2}{3} \right\rfloor -2r\right)\notag\\ &\phantom{xx}+ (1 - x_2)\,v(r,j)
\end{align}
completing the proof. Note that the constants $x_1$ and $x_2$ have different roles. For example, $x_1 \notin \{0,1\}$ guarantees that both the distance and the dimension are $\Theta(r^{j+2})$. The constant $x_2$, however, may be chosen to either increase the distance or the dimension.

\end{document}